\documentclass{llncs}
\usepackage[english]{babel}
\usepackage[utf8x]{inputenc}
\usepackage[T1]{fontenc}
\usepackage[a4paper,top=3cm,bottom=2.5cm,left=3cm,right=3cm,marginparwidth=1.75cm]{geometry}

\usepackage{enumitem}
\usepackage[skins]{tcolorbox}
\usepackage{amsmath}
\usepackage{mathtools}
\usepackage{graphicx}
\usepackage{amssymb}
\usepackage[colorinlistoftodos]{todonotes}
\usepackage[colorlinks=true, allcolors=blue,breaklinks=true]{hyperref}
\usepackage{breakcites}
\usepackage{upgreek}
\usepackage{bbold}
\usepackage{stmaryrd}
\usepackage{wasysym}
\usepackage{framed}
\usepackage{empheq}
\usepackage{enumitem}
\usepackage[boxed]{algorithm2e}
\usepackage{xcolor}
\usepackage{theorem}
\newcounter{counter}
\DeclareMathOperator*{\E}{\mathbb{E}}

\usepackage{float}
\usepackage{xspace}

\newtheorem{defi}[counter]{Definition}
\newtheorem{thm}[counter]{Theorem}
\newtheorem{lem}[counter]{Lemma}

\newtheorem{cor}[counter]{Corollary}

\theorembodyfont{\normalfont}
\newtheorem{rem}[counter]{Remark}

\newcommand{\bra}[1]{\langle #1|}
\newcommand{\ket}[1]{|#1\rangle}

\definecolor{dgreen}{rgb}{.1,.5,.1}
\definecolor{grey}{rgb}{.4,.4,.4}

\newcommand{\proj}[1]{\ket{#1}\!\bra{#1}}

\usepackage{textgreek}
\newcommand{\sigp}{\textSigma-protocol\xspace}
\newcommand{\sigps}{\textSigma-protocols\xspace}

\def\regZ{\ensuremath{\textit{\textsf{Z}}}}

\def\regE{\ensuremath{\textit{\textsf{E}}}}
\def\regX{\ensuremath{\textit{\textsf{X}}}}

\newcommand{\Gen}{\ensuremath{\mathsf{Gen}}\xspace}
\newcommand{\Sign}{\ensuremath{\mathsf{Sign}}\xspace}
\newcommand{\Ver}{\ensuremath{\mathsf{Verify}}\xspace}

\newcommand*\widefbox[1]{\fbox{\hspace{2em}#1\hspace{2em}}}

\DeclareMathSymbol{\shortminus}{\mathbin}{AMSa}{"39}

\newcommand{\switch}[2]{#2}  

\title{The Measure-and-Reprogram Technique 2.0: Multi-Round Fiat-Shamir and More 
}
\author{\vspace{-1cm}}\institute{}
\author{Jelle Don\inst{1} \and  Serge Fehr\inst{1,2} \and Christian Majenz\inst{1,3} }
\institute{
	Centrum Wiskunde \& Informatica (CWI), Amsterdam, Netherlands \and 
	Mathematical Institute, Leiden University, Netherlands \and
	QuSoft, Amsterdam, Netherlands 
	\\ \email{jelle.don@cwi.nl}, \email{serge.fehr@cwi.nl}, \email{c.majenz@uva.nl}}

\begin{document}
	\maketitle	
	\setcounter{footnote}{0}

	\begin{abstract}

    We revisit recent works by Don, Fehr, Majenz and Schaffner and by Liu and Zhandry on the security of the Fiat-Shamir transformation of $\Sigma$-protocols in the quantum random oracle model (QROM). Two natural questions that arise in this context are: (1) whether the results extend to the Fiat-Shamir transformation of {\em multi-round} interactive proofs, and (2) whether Don et al.'s $O(q^2)$ loss in security is optimal.
    
    Firstly, we answer question (1) in the affirmative. As a byproduct of solving a technical difficulty in proving this result, we slightly improve the result of Don et al., equipping it with a cleaner bound and an even simpler proof. We apply our result to digital signature schemes showing that it can be used to prove strong  security for schemes like MQDSS in the QROM. As another application we prove QROM-security of a non-interactive OR proof by Liu, Wei and Wong. 
    
    As for question (2), we show via a Grover-search based attack that Don et al.'s quadratic security loss for the Fiat-Shamir transformation of $\Sigma$-protocols is optimal up to a small constant factor. This extends to our new multi-round result, proving it tight up to a factor that depends on the number of rounds only, i.e. is constant for any constant-round interactive proof.  
	\end{abstract}

	\section{Introduction}
	
	\subsubsection{Reprogramming the quantum random oracle. }
\def\eps{\varepsilon}
We reconsider the recent work of Don, Fehr, Majenz and Schaffner~\cite{DFMS19} on the quantum random oracle model (QROM). On a technical level, they showed how to reprogram the QROM adaptively at {\em one} input. More precisely, for any oracle quantum algorithm ${\cal A}^H$, making $q$ calls to a random oracle $H$ and outputting a pair $(x,z)$ so that some predicate $V(x,H(x),z)$ is satisfied, they showed existence of a ``simulator'' $\cal S$ that mimics the random oracle, extracts $x$ from ${\cal A}^H$ by measuring one of the oracle queries to $H$, and then reprograms $H(x)$ to a given value $\Theta$ so that $z$ output by ${\cal A}^H$ now satisfies $V(x,\Theta,z)$, except with a multiplicative $O(q^2)$ loss in probability (plus a negligible additive loss). We emphasize that the challenging aspect of this problem is that ${\cal A}^H$'s queries to $H$ may be in quantum superposition, and thus measuring such a query disturbs the state and thus the behavior of ${\cal A}^H$. Still, Don et al. managed to control this disturbance sufficiently. In independent work and using very different techniques, Liu and Zhandry~\cite{LZ19} showed a similar kind of result, but with a $O(q^9)$ loss. 

As an immediate application of this technique, it is then concluded that the Fiat-Shamir transformation of a $\Sigma$-protocol is as secure (in the QROM) as the original $\Sigma$-protocol (in the standard model), up to a $O(q^2)$ loss, i.e., any of the typically considered security notions is preserved under the Fiat-Shamir transformation, even in the quantum setting. In combination with prior work on simulating signature queries \cite{Unruh2017,Kiltz2017}, security (in the QROM) of Fiat-Shamir signatures that arise from ordinary $\Sigma$-protocols then follows as a corollary. 

Given important examples of {\em multi-round} public-coin interactive proofs, used in, e.g.,  MQDSS \cite{MQDSS} and for Bulletproofs~\cite{Bulletproofs}%
\footnote{The security of the original Bulletproofs protocol relies on the hardness of discrete-log; however, work in progress considers post-quantum secure versions~\cite{PQBP-talk}.},
 a natural question that arises is whether these techniques and results extend to the reprogrammability of the QROM at {\em multiple} inputs and the security of the Fiat-Shamir transformation (in the QROM) of {\em multi-round} public-coin interactive proofs. Another question is whether the $O(q^2)$ loss (for the original $\Sigma$-protocols) is optimal, or whether one might hope for a linear loss as in the classical case. 

In this work, we provide answers to both these natural questions\,---\,and more. 

\subsubsection{A technical hurdle for generalizing \cite{DFMS19} to multi-round Fiat-Shamir. }
To start with, we observe that the naive approach of applying the original result of \cite{DFMS19} inductively so as to reprogram multiple inputs one by one does not work
. This is due to a subtle technical issue that has to do with the precise statement of the original result. In more detail, the statement involves an additive error term $\eps_x \geq 0$ that depends on the particular choice of the point $x$, which is (adaptively) chosen to be the input on which the random oracle (RO) is reprogrammed. The guarantee provided by \cite{DFMS19}  is that this error term stays negligible even {\em when summed over} all $x$'s, i.e., $\sum_x \eps_x = negl$. The formulation of the result for individual $x$'s with control over $\sum_x \eps_x$ is important for the later applications to the Fiat-Shamir transformation. 
However, when applying the result twice in a row, with the goal being to reprogram the RO at two inputs $x_1,x_2$, then we end up with two error terms $\eps_{x_1}$ and $\eps^{x_1}_{x_2}$ (with the second one depending on $x_1$), where the first one stays negligible when summed over $x_1$ and the second one stays negligible when summed over $x_2$ (for any $x_1$); but it is unclear that the sum $\eps_{x_1,x_2} := \eps_{x_1} + \eps^{x_1}_{x_2}$ stays negligible when summed over $x_1$ {\em and} $x_2$, which is what we would need to get the corresponding generalized statement.

\subsubsection{Our results }
As a first contribution, we revise the {\em original} result from \cite{DFMS19} of reprogramming the QROM at one input by showing an {\em improved} version that has {\em no} additive error term, but only the original multiplicative $O(q^2)$ loss. For typical direct cryptographic applications, this improvement makes no big quantitative difference due to the error term being negligible, but: (1) it makes the statement cleaner and easier to formulate, (2) somewhat surprisingly, the proof is simpler than that of the original result in \cite{DFMS19}, and (3) most importantly, it removes the technical hurdle to extend to multiple inputs. Indeed, we then get the desired multi-input reprogrammability result by means of a not too difficult, though somewhat tedious, induction argument. 

Building on our multi-input reprogrammability result above, our next goal then is to show the security of the Fiat-Shamir transformation (in the QROM) of multi-round public-coin interactive proofs. In contrast to the original result in [DFMS19] for the Fiat-Shamir transformation of Σ-protocols some additional work is needed here, to deal with the order of the messages extracted from the Fiat-Shamir adversary. Thus, as a stepping stone, we consider and analyze a variant of the above multi-input reprogrammability result, which enforces the right order of the extracted messages. As a simple corollary of this, we then obtain the desired security of multi-round Fiat-Shamir. Here, the multiplicative loss becomes $O(q^{2n})$ for a $(2n+1)$-round public-coin interactive proof with constant $n$.

In the context of digital signatures, the original motivation for the Fiat-Shamir transformation, we extend previous results by Unruh \cite{Unruh2017} and Don et al. \cite{DFMS19} to show that Fiat-Shamir signature schemes based on a multi-round,  honest-verifier zero knowledge public-coin interactive quantum proof of knowledge have standard signature security (existential unforgeability under chosen message attacks, UF-CMA) in the QROM. Assuming the additional collision-resistance-like property of computationally unique responses, they are even strongly unforgeable. We go on to apply this result to the signature scheme MQDSS \cite{MQDSS}, a candidate in the ongoing NIST standardization process for post-quantum cryptographic schemes \cite{NIST}, providing its first QROM proof.

Another application of our multi-round Fiat-Shamir result would for instance be to Bulletproofs~\cite{Bulletproofs}.

As a second application of our multi-input reprogrammability result, we show security (in the QROM) of the non-interactive OR-proof introduced by Liu, Wei and Wong \cite{LWW04}, further analyzed by Fischlin, Harasser and Janson \cite{FHJ20}. While the well-known (interactive) OR-proof by Cramer, Damg\r{a}rd and Schoenmakers \cite{CDS94} is a $\Sigma$-protocol and thus the results from \cite{DFMS19} apply, the inherently non-interactive OR-proof by Liu et al. does {\em not} follow this blueprint of being obtained as the Fiat-Shamir transformation of a $\Sigma$-protocol (though in some sense it is ``close'' to being of this form).  
We show here how the $2$-input version of our multi-input reprogrammability result implies security of this OR-proof in the QROM. 

Our last contribution is a lower bound that shows that the multiplicative $O(q^2)$ loss in the security argument of the Fiat-Shamir transformation of $\Sigma$-protocols is tight (up to a factor $4$). Thus, the $O(q^2)$ loss is unavoidable in general. Furthermore, we extend this lower bound to the Fiat-Shamir transformation of multi-round interactive proofs as considered in this work, and we show that also here to obtained loss $O(q^{2n})$ is in general optimal, up to a constant that depends on $n$ only. 

\subsubsection{Related work}
Before the recently obtained reduction \cite{DFMS19,LZ19} was available, the Fiat-Shamir tranform in the QROM was studied in a number of works \cite{Unruh2017,Dagdelen,Kiltz2017}, where weaker security properties were shown. In addition, Unruh developed an alternative transform \cite{Unruh2015} that provided QROM security at the expense of an increased proof size. The Unruh transform was later generalized to apply to 5-round public coin interactive proof systems \cite{SOFIA}.

	\section{Notation}

	Up to some modifications, we follow closely the notation used in~\cite{DFMS19}. 
	We consider a (purified) oracle quantum algorithm $\cal A$ that makes $q$ queries to an {\em oracle}, i.e., an unspecified function $H: {\cal X} \to {\cal Y}$ with finite non-empty sets ${\cal X},{\cal Y}$.   
	Formally, $\cal A$ is described by a sequence of unitaries $A_1,\ldots,A_q$ and an initial state $\ket{\phi_0}$.%
	\footnote{Alternatively, we may regard $\ket{\phi_0}$, as an additional input given to $\cal A$. }
	For technical reasons that will become clear later, we actually allow (some of) the $A_i$'s to be a {\em projection} followed by a unitary (or vice versa). One can think of such a projection as a measurement performed by the algorithm, with the algorithm aborting except in case of a particular measurement outcome. 	
	
	For any concrete choice of $H: {\cal X} \to {\cal Y}$, the algorithm $\cal A$ computes the state
	$$
	\ket{\phi_q^H} := {\cal A}^H \ket{\phi_0} := A_q\mathcal{O}^H \cdots A_1\mathcal{O}^H \ket{\phi_0} \, ,
	$$ 
	where $\mathcal{O}^H$ is the unitary defined by $\mathcal{O}^H : \ket{c}\ket{x}\ket{y} \mapsto \ket{c}\ket{x}\ket{y \oplus c \!\cdot\! H(x)}$ for any triple $c \in \{0,1\}$, $x \in \cal X$ and $y \in \cal Y$, with $\mathcal{O}^H$ acting on appropriate registers. 
	We emphasize that we allow {\em controlled} queries to $H$. Per se, this gives the algorithm more power, and thus will make our result only stronger, but it is easy to see that such controlled queries to the standard quantum oracle for a function can always be simulated by means of ordinary queries, at the price of one additional query.\footnote{Allowing controlled queries to the random oracle is also the more natural model compared to restricting to plain access to the unitary. After all, the motivation for the QROM is that in the real world, an attacker can implement the modeled hash function on their quantum computer, so they can definitely implement the controlled version as well.}
	The final state ${\cal A}^H \ket{\phi_0}$ is considered to be a state over registers $\regX = \regX_1\ldots\regX_n$, $\regZ$ and $\regE$. 
	
	Following~\cite{DFMS19}, we introduce the following notation. For $0 \leq i,j \leq q$ we set 
	$$
	\mathcal{A}_{i\rightarrow j}^H := A_{j}\mathcal{O}^H \cdots A_{i+1}\mathcal{O}^H \, ,
	$$
	where, by convention, $\mathcal{A}_{i\rightarrow j}^H$ is set to $\mathbb{1}$ if $j \leq i$. Furthermore, we let
	$$
	\ket{\phi_i^H} := \big(\mathcal{A}_{0\rightarrow i}^H\big)\ket{\phi_0} 
	$$ 
	be the state of $\cal A$ after the $i$-th step but right before the $(i+1)$-st query, which is consistent with $\ket{\phi_q^H}$ above. 
	
	For a given function $H: {\cal X} \to {\cal Y}$ and for fixed $x \in {\cal X}$ and $\Theta \in {\cal Y}$, we define the {\em reprogrammed} function $H\!*\!\Theta x: {\cal X} \to {\cal Y}$ that coincides with $H$ on ${\cal X} \setminus \{x\}$ but maps $x$ to $\Theta$. With this notation at hand, we can then write
	$$
	\big(\mathcal{A}_{i\rightarrow q}^{H*\Theta x}\big) \, \big(\mathcal{A}_{0\rightarrow i}^{H}\big) \, \ket{\phi_0} = \big(\mathcal{A}_{i\rightarrow q}^{H*\Theta x}\big)\ket{\phi_i^H}
	$$
	for an execution of $\cal A$ where the oracle is reprogrammed at a given point $x$ after the $i$-th query. We stress that $(\mathcal{A}_{i\rightarrow q}^{H*\Theta x}) (\mathcal{A}_{0\rightarrow i}^{H})$ can again be considered to be an oracle quantum algorithm $\cal B$, which depends on $\Theta \in {\cal Y}$, that makes $q$ queries to (the unprogrammed) function $H$. Indeed, the (controlled) queries to the reprogrammed oracle ${H*\Theta x}$ can be simulated by means of controlled queries to $H$ (using one additional ``work qubit'').%
	\footnote{Here it is crucial that we allow {\em controlled}  queries to $H$. }
	Exploiting that, in addition to unitaries, we allow projections as elementary operations, we can also understand $(\mathcal{A}_{i\rightarrow q}^{H*\Theta x}) X (\mathcal{A}_{0\rightarrow i}^{H})$ to be an oracle quantum algorithm again that makes oracle queries to $H$, where $X$ is the projection $X = \proj{x}$, acting on the corresponding query register to the oracle. 
	
	More generally, for any ${\mathbf x} = (x_1,\ldots,x_n)\in {\cal X}^n$ \emph{without duplicate entries}, i.e., $x_i \neq x_j$ for $i \neq j$, and for any ${{\mathbf \Theta}}\in {\cal Y}^n$, we define
	\begin{align*}
	&H*{\mathbf \Theta \mathbf x} = H*{\Theta_1 x_1}*\cdots*{\Theta_n x_n} : \, {\cal X} \to {\cal Y}\\
	&x \mapsto 
	\begin{cases}
	\Theta_i &\text{if $x = x_i$ for some $i \in \{1,\ldots,n\}$} \\
	H(x) &\text{otherwise}.
	\end{cases}
	\end{align*}
This will then allow us to consider $(\mathcal{A}_{i_2\rightarrow q}^{H*\Theta_1 x_1 * \Theta_2 x_2}) X_2 (\mathcal{A}_{i_1\rightarrow i_2}^{H*\Theta_1 x_1}) X_1 (\mathcal{A}_{0\rightarrow i_1}^{H})$ as an oracle quantum algorithm with oracle queries to $H$, etc.

	
	Eventually, we are interested in the probability that after the execution of the original algorithm ${\cal A}^H$, and upon measuring register $\regX$ in the computational basis to obtain ${\mathbf x} = (x_1,\ldots,x_n) \in {\cal X}^n$, the state of register $\regZ$ is of a certain form dependent on ${\mathbf x}$ and $H({\mathbf x}) = (H(x_1),\ldots,H(x_n))$.  
	Such a requirement (for a fixed ${\mathbf x}$) is captured by a projection 
	$$
	G_{\bf x}^{H} = \proj{{\mathbf x}} \otimes \Pi_{{\mathbf x},H({\mathbf x})} \, ,
	$$
	where $\{\Pi_{{\mathbf x},{\mathbf \Theta}}\}_{{\mathbf x},{\mathbf \Theta}}$ is a family of projections with ${\mathbf x} \in {\cal X}^n$ and ${\mathbf \Theta} \in {\cal Y}^n$, and with the understanding that $\proj{{\mathbf x}}$ acts on $\regX$ and $\Pi_{{\mathbf x},H({\mathbf x})}$ on register $\regZ$. We refer to such a family of projections as a \emph{quantum predicate}. 
	We use $G_{\mathbf x}^{\mathbf \Theta}$ as a short hand for $G_{\mathbf x}^{H*{\mathbf \Theta}{\mathbf x}}$, and we write $G_x^H$ and $G_x^{\Theta}$ with $x \in \cal X$ and $\Theta \in \cal Y$ for the case $n = 1$.
	
	For an arbitrary but fixed ${\bf x}_\circ \in {\cal X}^n$, we are then interested in the probability
	$$
	\Pr\bigr[\,{\bf x}\!=\! {\bf x}_\circ \wedge V({\bf x},H({\bf x}),z) : ({\bf x},z) \leftarrow {\cal A}^H \,\bigl]\, = \bigl\|G_{{\bf x}_\circ}^H \ket{\phi_q^H}\bigr\|_2^2 \, .
	$$
	where the left hand side is our notation for this probability, where we understand ${\cal A}^H$ to be an algorithm that outputs the measured $\bf x$ together with the quantum state $z$ in register $\regZ$, and $V$ to be the quantum predicate specified by the projections $\Pi_{{\mathbf x},{\mathbf\Theta}}$. Correspondingly, $\Pr\bigr[ x\!=\!  x_\circ \wedge V(x,H(x),z) : (x,z) \leftarrow {\cal A}^H \bigl]\, = \|G_{x_\circ}^H \ket{\phi_q^H}\|_2^2$ for the $n=1$ case.

	\section{An improved single-input reprogramming result}\label{secmainresult}
	
	For the case $n=1$, Don et al.~\cite{DFMS19} show the existence of a black-box {\em simulator} $\cal S$ such that for any oracle quantum algorithm $\cal A$ as considered above with oracle access to a {\em uniformly random} $H$, it holds that 
	\begin{align}\begin{split}
		\Pr_\Theta\bigr[&x\!=\! x_\circ \wedge V(x,\Theta,z) : (x,z) \leftarrow \langle{\cal S}^{\cal A} , \Theta\rangle\bigl]
		\\ &
		\geq \frac{1}{2(q\!+\!1)(2q\!+\!3)} \Pr_H\bigl[x\!=\! x_\circ \wedge V(x,H(x),z) : (x,z) \leftarrow {\cal A}^{H} \bigr] - \eps_{x_\circ} \, ,
		\end{split}\label{eq:old}\end{align}%
    for any $x_\circ \in \cal X$, where the $\eps_{x_\circ}$'s are non-negative and their sum over $x_\circ \in \cal X$ is bounded by $1/(2q|{\cal Y}|)$, i.e., negligible whenever $|{\cal Y}|$ is superpolynomial. The notation $(x,z) \leftarrow \langle{\cal S}^{\cal A} , \Theta\rangle$ is to be understood in that in a first stage ${\cal S}^{\cal A}$ outputs $x$, and then on input $\Theta$ it outputs $z$. 
    At the core, Equation \eqref{eq:old} follows from Lemma~1 of~\cite{DFMS19} which shows that
  		\begin{align}\begin{split}
		\E_{\Theta,i,b}&\left[\big\|(\proj{x} \otimes \Pi_{x,\Theta}) \big(\mathcal{A}_{i+b\rightarrow q}^{H*\Theta x}\big)\big(\mathcal{A}_{i\rightarrow i+b}^{H}\big)X\ket{\phi_i^H}\big\|_2^2\right] \\
		&\geq  \frac{\E_{\Theta}\Bigl[\big\|(\proj{x} \otimes \Pi_{x,\Theta})\ket{\phi_q^{H*\Theta x}}\big\|_2^2\Bigr]}{2(q+1)(2q+3) } - \frac{\big\|X\ket{\phi_q^H}\big\|_2^2}{2(q+1)|{\cal Y}| } \, ,
		\end{split}
		\label{eq:oldtechn}
		\end{align}
    and from which the construction of $\cal S$ can be extracted. The bound \eqref{eq:old} on the ``success probability'' of $\cal S$ then follows from the observation that $\cal S$ can simulate the calls to $H$ and to $H\!*\!\Theta x$ by means of a $2(q\!+\!1)$-wise independent hash function, and that $H$ and $H\!*\!\Theta x$ are indistinguishable for random $H$ and $\Theta$. 
    
    In this section we show an improved variant of Equation \eqref{eq:old}, which avoids the additive error term $\eps_{x_\circ}$. While having negligible quantitative effect in typcial situations, it makes the statement simpler. In addition, as explained in the introduction, it circumvents a technical issue one encounters when trying to extend to the multi-input case. Furthermore, our improved version comes with a simpler proof.%
    \footnote{We thank Dominique Unruh for the idea that it might be possible to  avoid the additive error term, and for proposing an argument for achieving that, which inspired us to find the simpler argument we eventually used. }
	
	The approach is to avoid the additive error term in Equation \eqref{eq:oldtechn}. We achieve this by  slightly tweaking the simulator $\cal S$. From the technical perspective, while on the left hand side of Equation \eqref{eq:oldtechn} the expectation is over a random $i \in \{0,\ldots,q\}$, selecting one of the $q+1$ queries of $\cal A$ at random (where the $\regX$ register of the output state is considered to be a final query), and a random $b \in \{0,1\}$, our new version has syntactically the same left hand side, but with the expectation over a random pair $(i,b) \in (\{0,\ldots,q\!\shortminus\!1\}\times \{0,1\})\cup \{(q,0)\}$ instead. This allows us to absorb the additive error term into the success probability of the simulator. Furthermore, it holds for any {\em fixed} choice of $\Theta$ (and not only on average for a random choice).

%
%
	
	\begin{lem}
		\label{lem:mainresult}
		Let $\cal A$ be a $q$-query oracle quantum algorithm. Then, for any function $H: {\cal X}\rightarrow {\cal Y}$, any $x \in \cal X$ and $\Theta \in \cal Y$, and any projection $\Pi_{x,\Theta}$, it holds that
		\begin{align*}
		\E_{i,b}\left[\big\|(\proj{x} \otimes \Pi_{x,\Theta})\big(\mathcal{A}_{i+b\rightarrow q}^{H*\Theta x}\big)\big(\mathcal{A}_{i\rightarrow i+b}^{H}\big)X\ket{\phi_i^H}\big\|_2^2\right] &\!\geq\! \frac{\big\|(\proj{x} \otimes \Pi_{x,\Theta}) \ket{\phi_q^{H*\Theta x}}\big\|_2^2}{(2q+1)^2 } \, ,
		\end{align*}
		where the expectation is over uniform $(i,b) \in (\{0,\ldots,q\!\shortminus\!1\}\times \{0,1\})\cup \{(q,0)\} $.%
	\end{lem}
	
	This new version of Equation \eqref{eq:oldtechn} translates to a simulator $\mathcal{S}$ that works by running $\mathcal{A}$, but with the following modifications. First, one of the $q+1$ queries of $\mathcal{A}$ (also counting the final output in register~$\regX$) is measured, and the measurement outcome $x$ is output by (the first stage of) $\mathcal{S}$. We emphasize that the crucial difference to \cite{DFMS19} is that each of the $q$ actual queries is picked with probability \smash{$\frac{2}{2q+1}$}, while the final output is picked with probability \smash{$\frac{1}{2q+1}$}. 
	Then, very much as in \cite{DFMS19}, this very query of $\cal A$ is answered either using the original $H$ {\em or} using the reprogrammed oracle $H\!*\!\Theta x$, with the choice being made at random\footnote{If it is the final output that is measured then there is nothing left to reprogram, so no choice has to be made. }, while all the remaining queries of $\cal A$ are answered using oracle $H\!*\!\Theta x$. Finally, (the second stage of) $\cal S$ outputs whatever $\mathcal{A}$ outputs. 
	
	In line with Theorem~1 in~\cite{DFMS19}, i.e.~Equation \eqref{eq:old} above, we obtain the following result from Lemma~\ref{lem:mainresult}. 
		\begin{thm}[Measure-and-reprogram, single input]\label{thm:main} 
		Let ${\cal X}$ and ${\cal Y}$ be finite non-empty sets. 
		There exists a black-box  two-stage quantum algorithm $\cal S$ with the following property. 
		Let $\cal A$ be an arbitrary oracle quantum algorithm that makes $q$ queries to a uniformly random $H: {\cal X}\rightarrow {\cal Y}$ and that outputs some $x \in {\cal X}$ and a (possibly quantum) output~$z$. Then, the two-stage algorithm ${\cal S}^{\cal A}$ outputs some $x \in {\cal X}$ in the first stage and, upon a random $\Theta \in {\cal Y}$ as input to the second stage, a (possibly quantum) output~$z$, so that for any \mbox{$x_\circ\in {\cal X}$} and any (possibly quantum) predicate $V$:
		\begin{align*}
		\Pr_\Theta\bigr[x\!=\! x_\circ& \wedge V(x,\Theta,z) : (x,z) \leftarrow \langle{\cal S}^{\cal A} , \Theta\rangle\bigl] 
		\switch{\;}{\\ &}
		\geq \frac{1}{(2q+1)^2} \Pr_H\bigl[x\!=\! x_\circ \wedge V(x,H(x),z) : (x,z) \leftarrow {\cal A}^{H} \bigr] \, .
		\end{align*}%
		Furthermore, $\cal S$ runs in time polynomial in $q$, $\log|\mathcal X|$ and $\log|\mathcal Y|$.
	\end{thm}
	The proof of Lemma~\ref{lem:mainresult} follows closely the proof of Equation \eqref{eq:old} in~\cite{DFMS19}, but the streamlined statement and simulator allow to cut some corners. 
	
	\begin{proof}[of Lemma~\ref{lem:mainresult}]
		For any $0\leq i \leq q$, inserting a resolution of the identity and exploiting that
		$$
		\big(\mathcal{A}_{i+1\rightarrow q}^{H*\Theta x}\big)\big(\mathcal{A}_{i\rightarrow i+1}^H\big)\big(\mathbb{1}-X\big)\ket{\phi_i^H} = \big(\mathcal{A}_{i\rightarrow q}^{H*\Theta x}\big)\big(\mathbb{1}-X\big)\ket{\phi_i^H} \, ,
		$$ 
		we can write
		\switch{
			\begin{align*}
			\big(\mathcal{A}_{i+1\rightarrow q}^{H*\Theta x}\big)\ket{\phi_{i+1}^H} 
			&=  \big(\mathcal{A}_{i+1\rightarrow q}^{H*\Theta x}\big)\big(\mathcal{A}_{i\rightarrow i+1}^H\big)\big(\mathbb{1}-X\big)\ket{\phi_i^H} \hspace{-8ex}& +\: \big(\mathcal{A}_{i+1\rightarrow q}^{H*\Theta x}\big)\big(\mathcal{A}_{i\rightarrow i+1}^H\big)X\ket{\phi_i^H} \notag\\
			&= \big(\mathcal{A}_{i\rightarrow q}^{H*\Theta x}\big)\big(\mathbb{1}-X\big)\ket{\phi_i^H} \hspace{-8ex}\!\!&\!\! +\: \big(\mathcal{A}_{i+1\rightarrow q}^{H*\Theta x}\big)\big(\mathcal{A}_{i\rightarrow i+1}^H\big)X\ket{\phi_i^H}\\
			&= \big(\mathcal{A}_{i\rightarrow q}^{H*\Theta x}\big)\ket{\phi_i^H} - \big(\mathcal{A}_{i\rightarrow q}^{H*\Theta x}\big)X\ket{\phi_i^H} \hspace{-8ex}\!\!&\!\! +\: \big(\mathcal{A}_{i+1\rightarrow q}^{H*\Theta x}\big)\big(\mathcal{A}_{i\rightarrow i+1}^H\big)X\ket{\phi_i^H}
			\end{align*}
		}{
			\begin{align*}
			&\big(\mathcal{A}_{i+1\rightarrow q}^{H*\Theta x}\big)\ket{\phi_{i+1}^H} &\\
			&\hspace{30pt}=  \big(\mathcal{A}_{i+1\rightarrow q}^{H*\Theta x}\big)\big(\mathcal{A}_{i\rightarrow i+1}^H\big)\big(\mathbb{1}-X\big)\ket{\phi_i^H} \!\!&\!\! +\: \big(\mathcal{A}_{i+1\rightarrow q}^{H*\Theta x}\big)\big(\mathcal{A}_{i\rightarrow i+1}^H\big)X\ket{\phi_i^H} \notag\\
			&\hspace{30pt}= \big(\mathcal{A}_{i\rightarrow q}^{H*\Theta x}\big)\big(\mathbb{1}-X\big)\ket{\phi_i^H} \!\!&\!\! +\: \big(\mathcal{A}_{i+1\rightarrow q}^{H*\Theta x}\big)\big(\mathcal{A}_{i\rightarrow i+1}^H\big)X\ket{\phi_i^H}\\
			&\hspace{30pt}= \big(\mathcal{A}_{i\rightarrow q}^{H*\Theta x}\big)\ket{\phi_i^H} - \big(\mathcal{A}_{i\rightarrow q}^{H*\Theta x}\big)X\ket{\phi_i^H} \!\!&\!\! +\: \big(\mathcal{A}_{i+1\rightarrow q}^{H*\Theta x}\big)\big(\mathcal{A}_{i\rightarrow i+1}^H\big)X\ket{\phi_i^H}
			\end{align*}
		}
		Rearranging terms, applying $G_{x}^\Theta = (\proj{x} \otimes \Pi_{x,\Theta})$ and using the triangle equality, we can thus bound
		\begin{align*}
		\big\| G_{x}^\Theta  \big(\mathcal{A}_{i\rightarrow q}^{H*\Theta x}\big) \ket{\phi_i^H} \big\|_2 
		\leq \big\| G_{x}^\Theta & \big(\mathcal{A}_{i+1\rightarrow q}^{H*\Theta x}\big)\ket{\phi_{i+1}^H}\big\|_2 \\& + \big\| G_{x}^\Theta \big(\mathcal{A}_{i\rightarrow q}^{H*\Theta x}\big)X\ket{\phi_i^H}\big\|_2\\ & \qquad + \big\| G_{x}^\Theta \big(\mathcal{A}_{i+1\rightarrow q}^{H*\Theta x}\big)\big(\mathcal{A}_{i\rightarrow i+1}^H\big)X\ket{\phi_i^H}\big\|_2 \, .
		\end{align*}
		Summing up the respective sides of the inequality over $i=0,\ldots,q-1$, we get
		\begin{equation*}
		\big\| G_x^\Theta\ket{\phi_{q}^{H*\Theta x}}\big\|_2 \:\leq\: \big\|  G_x^\Theta\ket{\phi_{q}^H}\big\|_2 + \!\!\!\sum_{\substack{0\leq i < q \\ b\in \{0,1\}}}\!\!\! \big\| G_x^\Theta \big(\mathcal{A}_{i+b\rightarrow q}^{H*\Theta x}\big)\big(\mathcal{A}_{i\rightarrow i+b}^H\big)X\ket{\phi_i^H}\big\|_2 \, .\label{eqnsquarthis}
		\end{equation*}
		By squaring both sides, dividing by $2q+1$ (i.e., the number of terms on the right hand side), and using Jensen's inequality on the right hand side, we obtain 
		$$
		\frac{\big\| G_x^\Theta\ket{\phi_{q}^{H*\Theta x}}\big\|_2^2}{2q+1} \leq \big\| G_x^\Theta\ket{\phi_{q}^H}\big\|_2^2 + \!\!\!\sum_{\substack{0\leq i < q \\ b\in \{0,1\}}}\!\!\!\big\| G_x^\Theta \big(\mathcal{A}_{i+b\rightarrow q}^{H*\Theta x}\big)\big(\mathcal{A}_{i\rightarrow i+b}^H\big)X\ket{\phi_i^H}\big\|_2^2  
		$$
		and thus, noting that we can write $\big\| G_x^\Theta\ket{\phi_{q}^H}\big\|_2^2$ as 
		$$\big\|G_x^\Theta \big(\mathcal{A}_{i+b\rightarrow q}^{H*\Theta x}\big)\big(\mathcal{A}_{i\rightarrow i+b}^H\big)X\ket{\phi_i^H}\big\|_2^2$$
		 with $i=q$ and $b=0$,
		\begin{equation*}\label{eq:intermediate0}
		 \frac{\big\| G_x^\Theta\ket{\phi_{q}^{H*\Theta x}}\big\|_2^2}{(2q+1)^2}\;\leq\; \E_{i,b}\left[\big\| G_x^\Theta \big(\mathcal{A}_{i+b\rightarrow q}^{H*\Theta x}\big)\big(\mathcal{A}_{i\rightarrow i+b}^H\big)X\ket{\phi_i^H}\big\|_2^2\right] \, .
		\end{equation*}
	\qed\end{proof}

For completeness, let us spell out how Theorem~8 of~\cite{DFMS19} on the generic security of the Fiat-Shamir transformation (in the QROM) can now be re-phrased, avoiding the negligible error term present in \cite{DFMS19}. We refer to~\cite{DFMS19} or to our later Section~\ref{sec:mFS} for the details on the Fiat-Shamir transformation. 

		\begin{thm}\label{thm:FS} 
		There exists a black-box quantum polynomial-time two-stage quantum algorithm $\cal S$ such that for any adaptive Fiat-Shamir adversary $\cal A$, making $q$ queries to a uniformly random function $H$ with appropriate domain and range, and for any $x_\circ \in {\cal X}$: 
		\begin{align*}
		\Pr\bigr[x\!=\! x_\circ \wedge v = accept& :(x,v) \leftarrow \langle{\cal S}^{\cal A} , {\cal V}\rangle\bigl] 
		\switch{\;}{\\ &}
		\geq \frac{1}{(2q+1)^2} \Pr_H\bigr[x\!=\! x_\circ \wedge V^H_{FS}(x,\pi) : (x,\pi) \leftarrow {\cal A}^H \bigl] \, . 
		\end{align*} 
	\end{thm}
	%

\section{Multi-input reprogrammability}\label{subsec:tech-n}
In this section, we extend our (improved) results on adaptively reprogramming the quantum random oracle at {\em one} point $x \in {\cal X}$ to {\em multiple} points $x_1,\ldots,x_n \in {\cal X}$. This in turn will allow us to extend the results on the security of the Fiat-Shamir transformation to {\em multi-round} protocols. We point out again that the improvement of  Lemma \ref{lem:mainresult} over Lemma 1 in ~\cite{DFMS19} plays a crucial role here, in that it circumvents the trouble with the negligible error term that occurs when trying to extend the result from~\cite{DFMS19} to the setting considered here. 

The starting point is the following generalized version of the problem considered in Section~\ref{secmainresult}. We assume an oracle quantum algorithm ${\cal A}^H$ that makes $q$ queries to a random oracle $H: {\cal X}\rightarrow {\cal Y}$ and then produces an output of the form $(x_1,\ldots,x_n,z)$, where $z$ may be quantum, such that a certain (quantum) predicate $V(x_1, H(x_1),\ldots,x_n,H(x_n),z)$ is satisfied with some probability. 
The goal then is to turn such an ${\cal A}^H$ into a multi-stage quantum algorithm ${\cal S}$ (the {\em simulator}) that, stage by stage, outputs the $x_i$'s and takes corresponding $\Theta_i$'s as input, and eventually outputs a (possibly quantum) $z$ with the property that $V(x_1, \Theta_1,\ldots,x_n,\Theta_n,z)$ is satisfied with similar probability.

\subsection{The general case}

Naively, one might hope for an ${\cal S}$ that outputs $x_1$ in the first stage (obtained by measuring one of the queries of ${\cal A}^H$), and then on input $\Theta_1$ proceeds by outputting $x_2$ in the second stage (obtained by measuring one of the subsequent queries of ${\cal A}^H$), etc. However, since ${\cal A}^H$ may query the hashes of $x_1,\ldots,x_n$ in an arbitrary order, we cannot hope for this to work. 
Therefore, we have to allow $\cal S$ to produce $x_1,\ldots,x_n$ in an arbitrary order as well.%
\footnote{Looking ahead, in Section~\ref{sec:EnforceOrder} we will force ${\cal A}^H$ to query, and thus $\cal S$ to extract, $x_1,\ldots,x_n$ in the {\em right} order by requiring $x_2$ to contain $H(x_1)$ as a substring, $x_3$ to contain $H(x_2)$ as a substring, etc. This will be important for the the multi-round Fiat-Shamir application.  }
Formally, we consider $\cal S$ with the following syntactic behavior: in the first stage it outputs a permutation $\pi$ together with $x_{\pi(1)}$ and takes as input $\Theta_{\pi(1)}$, and then for every subsequent stage $1 < i \leq n$ it outputs $x_{\pi(i)}$ and takes as input $\Theta_{\pi(i)}$; eventually, in the final stage (labeled by $n+1$) it outputs $z$. In line with earlier notation, but taking this additional complication into account, we denote such an execution of $\cal S$ as $(\pi,\pi({\mathbf x}),z) \leftarrow \langle{\cal S}^{\cal A}, \pi({\mathbf \Theta})\rangle$.

A final issue is that if $x_i = x_j$ then $H(x_i) = H(x_j)$ as well, whereas $\Theta_i$ and $\Theta_j$ may well be different. Thus, we can only expect $\cal S$ to work well when $x_1,\ldots x_n$ has no duplicates.

For us to be able to mathematically reason about the simulator described above, we introduce some additional notation.
For the basic simulator from Lemma \ref{lem:mainresult} we write, using $r_1=(b_1,i_1)$, as
$$
\mathcal{S}^{H,{\cal A}}_{\Theta_1,x_1,r_1} := 
\mathcal{S}^{H,{\cal A},\Theta_1,x_1,r_1} := \big(\mathcal{A}_{i_1+b_1\rightarrow q}^{H*\Theta_1x_1}\big)\big(\mathcal{A}_{i_1\rightarrow i_1+b_1}^{H}\big)X_1\big(\mathcal{A}_{0\rightarrow i_1}^{H}\big) \, .
$$
This can be recursively extended by applying it to ${\cal A}^H$ now being $\mathcal{S}^{H,{\cal A}}_{\Theta_1,x_1,r_1}$ so as to obtain
$$
\mathcal{S}^{H,{\cal A}}_{\Theta_{1,2},x_{1,2},r_{1,2}} :=  \big(\mathcal{S}_{i_2+b_2\rightarrow q}^{H*\Theta_2x_2,{\cal A},\Theta_1,x_1,r_1}\big)\big(\mathcal{S}_{i_2\rightarrow i_2+b_2}^{H,{\cal A},\Theta_1,x_1,r_1}\big)X_2\big(\mathcal{S}_{0\rightarrow i_2}^{H,{\cal A},\Theta_1,x_1,r_1}\big).
$$
In general, we can consider the following operator, which simulates $\mathcal{A}$ and performs $n$ measurements: 
$$
\mathcal{S}_{{\mathbf \Theta},\mathbf x,{\mathbf r}}^{H,{\cal A}}  :=  \big(\mathcal{S}_{i_n+b_n\rightarrow q}^{H*\Theta_nx_n,{\cal A},{\overline{\mathbf \Theta}},{\overline{\mathbf x}},\overline{\mathbf r}}\big)\big(\mathcal{S}_{i_n\rightarrow i_n+b_n}^{H,{\cal A},{\overline{\mathbf \Theta}},{\overline{\mathbf x}},\overline{\mathbf r}}\big)X_n\big(\mathcal{S}_{0\rightarrow i_n}^{H,{\cal A},{\overline{\mathbf \Theta}},{\overline{\mathbf x}},\overline{\mathbf r}}\big).
$$
where, for arbitrary but fixed $n$ and ${\bf\Theta} = (\Theta_1,\ldots,\Theta_n) \in {\cal Y}^n$, the notation $\overline{\mathbf \Theta}$ is understood as $\overline{\mathbf \Theta} = (\Theta_1,\ldots,\Theta_{n-1}) \in {\cal Y}^{n-1}$, and correspondingly for $\mathbf x$ etc. 
Finally, when considering {\em fixed} ${\mathbf \Theta} \in {\cal Y}^n$ and ${\mathbf x} \in {\cal X}^n$, we write
$$
S_{\mathbf{r}}^H({\cal A}) := \mathcal{S}_{{\mathbf \Theta},\mathbf x,{\mathbf r}}^{H,{\cal A}} \, .
$$
At the core of our multi-round result will be the following technical lemma, which generalizes Lemma \ref{lem:mainresult}.
\begin{lem}\label{lem:mainresultmulti}
		Let $\cal A$ be a $q$-query oracle quantum algorithm. Then, for any function $H: {\cal X}\rightarrow {\cal Y}$, any ${\mathbf x} \in {\cal X}^n$ and ${\mathbf \Theta}^n\in {\cal Y}^n$, and any projection $\Pi_{{\mathbf x,\mathbf \Theta}}$, it holds that
	\begin{align*}
	&\frac{\big\|\big(\proj{\mathbf x}\otimes\Pi_{{\mathbf x, \mathbf \Theta}}\big){\cal A}^{H*{\mathbf \Theta \mathbf x}}\ket{\phi_0}\big\|_2^2}{(2q+1)^{2n}}\leq \E_{\mathbf r}\left[\big\|\big(\proj{\mathbf x}_A\otimes\Pi_{{\mathbf x, \mathbf\Theta}}\big)\mathcal{S}_{\mathbf r}^H({\cal A})\ket{\phi_0}\big\|_2^2\right].
	\end{align*}
\end{lem}
\begin{proof}
	The proof is by induction on $n$, where the base case is given by Lemma~\ref{lem:mainresult}.
	
	For the induction step we first apply the base case, substituting $x_n$ for $x_1$, $\Theta_n$ for $\Theta_1$, $r_n$ for $r_1$, $H\!*\!{\overline{\mathbf \Theta}\overline{\mathbf x}}$ for $H$, and $\hat{\Pi}_{x_n,\Theta_n}$ for $\Pi_{x_1,\Theta_1}$, where
	$$\hat{\Pi}_{x_n,\Theta_n} = \proj{x_1}\otimes\ldots\otimes\proj{x_{n\shortminus 1}}\otimes\Pi_{\mathbf x,\mathbf \Theta}
	$$
	to obtain
	\begin{align*}
	&\frac{\big\|\big(\proj{x_n}\otimes \hat{\Pi}_{x_n,\Theta_n}\big)\mathcal{A}^{\left(H*{\overline{\mathbf \Theta}\overline{\mathbf x}}\right)*\Theta_nx_n}\ket{\phi_0}\big\|_2^2}{(2q+1)^2}\\
	&\quad\ \ \leq  \E_{r_n}\left[\big\|\big(\proj{x_n}_A\otimes \hat{\Pi}_{x_n,\Theta_n}\big)\mathcal{S}_{r_n}^{H*{\overline{\mathbf \Theta}\overline{\mathbf x}}}({\cal A})\ket{\phi_0}\big\|_2^2\right]
	\end{align*}
	which we can write as
	\begin{align}\label{eqn:exprn}
	\frac{\big\|\big(\proj{\mathbf x}\otimes\Pi_{{\mathbf x,\mathbf \Theta}}\big)\mathcal{A}^{H*{\mathbf \Theta \mathbf x}}\ket{\phi_0}\big\|_2^2}{(2q+1)^{2n}}&\leq  \frac{\E_{r_n}\left[\big\|\big(\proj{\mathbf x}\otimes\Pi_{\mathbf x,\mathbf \Theta}\big)\mathcal{S}_{r_n}^{H*{\overline{\mathbf \Theta}\overline{\mathbf x}}}({\cal A})\ket{\phi_0}\big\|_2^2\right]}{(2q+1)^{2(n\shortminus 1)}}
	\end{align}
	dividing both sides by $(2q+1)^{2(n\shortminus 1)}$ and swapping registers appropriately (to make sure that the register which contains $x_n$ comes after the others).
	
	Now fix $r_n$. We define
	$$
	\hat{\Pi}_{{\overline{\mathbf x},\overline{\mathbf \Theta}}} := \proj{x_n}\otimes\Pi_{\mathbf x,\mathbf \Theta}.
	$$
	and apply the induction hypothesis for $n\!-\!1$, substituting ${\cal S}_{r_n}^{H*\overline{\mathbf \Theta} \overline{\mathbf x}}({\cal A})$ for ${\cal A}^{H*{\overline{\mathbf \Theta} \overline{\mathbf x}}}$, and $\hat{\Pi}_{{\overline{\mathbf x}},\overline{\mathbf \Theta}}$ for $\Pi_{{\overline{\mathbf x}},\overline{\mathbf \Theta}}$, in order to derive
	\begin{align*}
	\frac{\big\|\big(\proj{\mathbf x}\otimes\Pi_{\mathbf x,\mathbf \Theta}\big)\mathcal{S}_{r_n}^{H*{\overline{\mathbf \Theta}\overline{\mathbf x}}}({\cal A})\ket{\phi_0}\big\|_2^2}{(2q+1)^{2(n\shortminus 1)}} &= \frac{\big\|\big(\proj{\overline{\mathbf{x}}}\otimes\hat{\Pi}_{{\overline{\mathbf x}},\overline{\mathbf \Theta}}\big)\mathcal{S}_{r_n}^{H*{\overline{\mathbf \Theta}\overline{\mathbf x}}}({\cal A})\ket{\phi_0}\big\|_2^2}{(2q+1)^{2(n\shortminus 1)}} \\
	&\leq \E_{\overline{\mathbf r}}\left[\big\|\big(\proj{\overline{\mathbf{x}}}\otimes\hat{\Pi}_{{\overline{\mathbf x}},\overline{\mathbf \Theta}}\big)\mathcal{S}_{\overline{\mathbf r}}^{H}({\cal S}_{r_n}({\cal A}))\ket{\phi_0}\big\|_2^2\right]  \\
	&=\E_{\overline{\mathbf{r}}}\left[\big\|\big(\proj{\mathbf x}\otimes\Pi_{{\mathbf x,\mathbf \Theta}}\big)\mathcal{S}_{\mathbf r}^{H}({\cal A})\ket{\phi_0}\big\|_2^2\right].
	\end{align*}
	Since this inequality holds for any fixed $r_n$, it also holds in expectation over $r_n$. Substituting it in \mbox{Equation \ref{eqn:exprn}}, we retrieve the statement of the lemma. 
\qed\end{proof}

\begin{rem}\label{rem:disjointslots}
In case of ${\bf x} = (x_1,\ldots,x_n) \in {\cal X}^n$ {\em without duplicate entries}, it follows from the resulting mutual orthogonality of the projections $X_j$ and the definition of $\mathcal{S}_{\mathbf r}^H({\cal A})$ that the following holds. The term in the expectation $\E_{\mathbf r}$ in the inequality of Lemma~\ref{lem:mainresultmulti} vanishes for any ${\bf r} = ({\bf i},{\bf b})$ for which there exist two distinct coordinates $j \neq k$ with $i_j = i_k$. As such, we may well understand this expectation to be over ${\bf r} = ({\bf i},{\bf b})$ for which $i_j \neq i_k$ whenever $j \neq k$; this only increases the expectation.%
\footnote{One might try to exploit this actual improvement in the bound; however, for typical choices of parameters, with $n$ a small constant and $q$ large, this is insignificant. }
In other words, we may assume that random {\em distinct} queries are measured in order to extract $x_1,\ldots,x_n$. 
\end{rem}

\begin{thm}[Measure-and-reprogram, multiple inputs]\label{thm:multiplemar}
	Let $n$ be a positive integer, and let ${\cal X},{\cal Y}$ be finite non-empty sets. 
	There exists a black-box polynomial-time $(n\!+\!1)$-stage quantum algorithm $\cal S$ with the syntax as outlined at the start of this section, satisfying the following property. 
	Let $\cal A$ be an arbitrary oracle quantum algorithm that makes $q$ queries to a uniformly random $H: {\cal X}\rightarrow {\cal Y}$ and that outputs a tuple ${\mathbf x} \in {\cal X}^n$ and a (possibly quantum) output~$z$. Then,  
	for any $\mathbf x^\circ\in X^n$ \emph{without duplicate entries} and for any predicate	$V$:
	\begin{align*}
	\Pr_{{{\mathbf \Theta}}}\bigr[{\mathbf x}\!=\! \mathbf x^\circ& \wedge V({\mathbf x},{{\mathbf \Theta}},z) : (\pi,\pi({\mathbf x}),z) \leftarrow \langle{\cal S}^{\cal A} , \pi({\mathbf \Theta})\rangle\bigl] 
	\\ &
	\geq \frac{1}{(2q+1)^{2n}} \Pr_H\bigl[{\mathbf x}\!=\! \mathbf x^\circ \wedge V({\mathbf x},H(\mathbf{x}),z) : ({\mathbf x},z) \leftarrow {\cal A}^{H} \bigr] \, .
	\end{align*}%
\end{thm}

\begin{proof}
We consider the inequality of Lemma~\ref{lem:mainresultmulti} with the expectation over $\bf r$ understood as in Remark~\ref{rem:disjointslots}. 
Additionally taking the expectation over $H$ and ${\mathbf \Theta}$ on both sides, we obtain
	\begin{align*}
	&\E_{H,\mathbf{\Theta}}\left[\frac{\big\|\big(\proj{\mathbf x}\otimes\Pi_{{\mathbf x, \mathbf \Theta}}\big){\cal A}^{H*{\mathbf \Theta \mathbf x}}\ket{\phi_0}\big\|_2^2}{(2q+1)^{2n}}\right]\leq \E_{H,\mathbf{\Theta},\mathbf r}\left[\big\|\big(\proj{\mathbf x}\otimes\Pi_{{\mathbf x, \mathbf\Theta}}\big)\mathcal{S}_{\mathbf r}^H({\cal A})\ket{\phi_0}\big\|_2^2\right]
	\end{align*}
	and note that this is equivalent to 
	\begin{align*}
	&\E_{H}\left[\frac{\big\|\big(\proj{\mathbf x}\otimes\Pi_{{\mathbf x, H(\mathbf{x})}}\big){\cal A}^{H}\ket{\phi_0}\big\|_2^2}{(2q+1)^{2n}}\right]\leq \E_{H,\mathbf{\Theta},\mathbf r}\left[\big\|\big(\proj{\mathbf x}\otimes\Pi_{{\mathbf x, \mathbf\Theta}}\big)\mathcal{S}_{\mathbf r}^H({\cal A})\ket{\phi_0}\big\|_2^2\right].
	\end{align*}
	since all values $\Theta_j$ and $H(x_j)$ have the same distribution. The term $\mathcal{S}_{\mathbf r}^H({\cal A})\ket{\phi_0} = \mathcal{S}_{{\mathbf \Theta},\mathbf x,{\mathbf r}}^{H,{\cal A}}\ket{\phi_0}$ corresponds to the output of the simulator that uses oracle access to $H$ to run ${\cal A}$ on an initial state $\ket{\phi_0}$, while measuring queries $i_j$ (finding $x_j$ as the outcome) and reprogramming the oracle at $x_j$ to $\Theta_j$ from the $(i_j+b_j)$-th query onwards, with $(i_j,b_j)=r_j$. 
	
	Next, we note that the value of the right hand side does not change \cite{Zhandry2012a} when instead of giving ${\cal S}$ oracle access to $H$, we let it choose a random instance from a family of $2q$-wise\footnote{It is easy to see that the result of \cite{Zhandry2012a} also holds for controlled-query algorithms. Alternatively, the $q$ controlled queries can be simulated using $q+1$ plain queries, and a $2(q+1)$-wise independent function can be used.} independent hash functions to simulate ${\cal A}$ on.
	  The choice of ${\mathbf r}$ uniquely determines the permutation $\pi$ with the property \smash{$i_{\pi(1)} < \cdots < i_{\pi(n)}$}; by definition of $\mathcal{S}_{{\mathbf \Theta},\mathbf x,{\mathbf r}}^{H,{\cal A}}$, the values ${\mathbf x} = (x_1,\ldots, x_n)$ are then extracted from the adversary's queries in the order \mbox{$\pi(\mathbf x) = (x_{\pi(1)},\ldots, x_{\pi(n)})$}. 
	 Since ${\cal S}$ chooses this $\mathbf{r}$ itself, we can assume that it includes $\pi$ in its output. Likewise, the simulator takes as input to every stage\,---\,from the second to the \mbox{$(n\!+\!1)$-st}\,---\,a fresh random value, in the order given by $\pi(\mathbf{\Theta})$. However, by definition of $\Pi_{{\mathbf x, \mathbf \Theta}}$ the final output of the simulator satisfies the predicate $V$ with respect to the given order (without $\pi$), i.e. such that $V(\mathbf x,{\mathbf \Theta},z) = 1$, as is the claim of the theorem. 
\qed\end{proof}

\subsection{The time-ordered case}\label{sec:EnforceOrder}

In some applications, like the multi-round version of the Fiat-Shamir transformation, we need that the simulator extracts the messages in the right order. This can be achieved by replacing the hash {\em list} $H({\bf x}) = \big(H(x_1),\ldots,H(x_n)\big)$, consisting of individual hashes, by a hash {\em chain}, where subsequent hashes depend on previous hashes. Intuitively, this enforces $\cal A$ to query the oracle in the given order. 

%

Formally, considering a function $H: ({\cal X}_0 \cup {\cal Y}) \times {\cal X}\rightarrow {\cal Y}$ and given a tuple ${\mathbf x} = (x_0,x_1,\ldots,x_n)$ in ${\cal X}_0 \times {\cal X}^n$, we define the {\em hash chain} $\mathbf{h}^{H,\mathbf{x}} = \big(h_1^{H,\mathbf{x}},\ldots,h_n^{H,\mathbf{x}}\big)$ given by 
$$
h_1^{H,\mathbf{x}}=H(x_0,x_1) 
\qquad\text{and}\qquad 
h_i^{H,\mathbf{x}} := H\big(h_{i-1}^{H,\mathbf{x}},x_i\big) 
$$
for $2\leq i\leq n$. 

\begin{thm}[Measure-and-reprogram, enforced extraction order]\label{thm:enforced}
	Let $n$ be a positive integer, and let ${\cal X}_0,{\cal X}$ and ${\cal Y}$ be finite non-empty sets. 
	There exists a black-box polynomial-time $(n\!+\!1)$-stage quantum algorithm ${\cal S}$, satisfying the following property.
	Let $\cal A$ be an arbitrary oracle quantum algorithm that makes $q$ queries to a uniformly random $H: ({\cal X}_0\cup {\cal Y}) \times {\cal X}\rightarrow {\cal Y}$ and that outputs a tuple ${\mathbf x} = (x_0,x_1,\ldots,x_n) \in \left({\cal X}_0\times{\cal X}^n\right)$ and a (possibly quantum) output~$z$.
	Then, for any $\mathbf{x}^\circ\in ({\cal X}_0\times {\cal X}^n)$ without duplicate entries and for any predicate $V$:
	\begin{align*}
	&\Pr_{{{\mathbf \Theta}}}\bigr[\mathbf{x}\!=\! \mathbf{x}^\circ \wedge V({\mathbf x},{{\mathbf \Theta}},z) : ({\mathbf x},z) \leftarrow \langle{{\cal S}^A} , {\mathbf	 \Theta}\rangle\bigl] 
	\\ &
	\geq \frac{n!}{(2q+n+1)^{2n}} \Pr_H\bigl[\mathbf{x}\!=\! \mathbf{x}^\circ\wedge V({\mathbf x},\mathbf{h}^{H,\mathbf{x}},z) : ({\mathbf x},z) \leftarrow {\cal A}^{H} \bigr]-\epsilon_{\mathbf x^\circ}\, .
	\end{align*}%
	where $\epsilon_{\mathbf x^\circ}$ is equal to $\frac{n!}{|{\cal Y}|}$ when summed over all $\mathbf{x^\circ}$.
\end{thm}

\begin{rem}\label{rem:decrease-additive-error}
	The additive error term $n!/|{\cal Y}|$ stems from the fact that the extraction in the right order fails if $\cal A$ succeeds in guessing one (or more) of the hashes in the hash chain. The claimed term can be improved to $(n-1)^2/|{\cal Y}| + n!/|{\cal Y}|^2$ by doing a more fine-grained analysis, distinguishing between permutations $\pi \neq \mathrm{id}$ that bring 2 elements ``out of order'' or more. In any case, it can be made arbitrary small by extending the range $\cal Y$ of $H$ for computing the hash chain. 
	
\end{rem}

\begin{proof}
	First, we note that $V({\mathbf x},\mathbf{h}^{H,\mathbf{x}},z)= V'(\mathbf{v},H(\mathbf{v}),z)$ for ${\bf v} = (v_1,\ldots,v_n)$ given by $v_1 = (x_0,x_1)$ and $v_i = \big(h_{i-1}^{H,\mathbf{x}},x_i\big) = \big(H(v_{i-1}),x_i\big)$ for $i\geq 2$, and  $V'(\mathbf{v},\mathbf{h},z) := \big[\,V(\mathbf{x},\mathbf{h},z) \,\wedge\, h'_{i} \!=\! h_{i-1} \forall i\geq 2 \,\big]$ for any $\bf v$ of the form $v_1 = (x_0,x_1)$ and $v_i = \big(h'_i,x_i\big)$ for $i\geq 2$. Next, at the cost of $n$ additional queries, we can extend $\cal A$ to an algorithm ${\cal A}_+$ that actually outputs $(\mathbf v,z)$, since ${\cal A}_+$ can easily obtain the $H(v_i)$'s by making $n$ queries to $H$. These observations together give
	\begin{align*}
	\Pr_H\bigl[\mathbf{x}\!=\! \mathbf x^\circ& \wedge V({\mathbf x},\mathbf{h}^{H,\mathbf{x}},z) : ({\mathbf x},z) \leftarrow {\cal A}^{H} \bigr]
	=
	\Pr_H\bigl[\mathbf x\!=\! \mathbf x^\circ \wedge V'(\mathbf{v},H(\mathbf{v}),z) : ({\mathbf v},z) \leftarrow {\cal A}_+^{H} \bigr] \, .
	\end{align*}
	
	Let $\mathbf v^\circ = (v_1^\circ,\ldots,v_n^\circ)$ with $v_i^\circ := (h^\circ_i,x^\circ_i)$, where $h_1^\circ = x^\circ_0$ and $h_i^\circ \in {\cal Y}$ is arbitrary but fixed for $i \geq 2$. Let $\mathbf\Theta$ be uniformly random in ${\cal Y}^n$. An application of Theorem \ref{thm:multiplemar} yields a simulator $\hat{\cal S}$ with
	\begin{align*}
	\Pr_{{{\mathbf \Theta}}}\bigr[{\mathbf v}\!=\! \mathbf v^\circ& \wedge V'({\mathbf v},{{\mathbf \Theta}},z) : (\pi,\pi({\mathbf v}),z) \leftarrow \langle{\hat{\cal S}}^{\cal A_+} , \pi({\mathbf \Theta})\rangle\bigl] 
	\\ &
	\geq \frac{1}{(q+n+1)^{2n}} \Pr_H\bigl[{\mathbf v}\!=\! \mathbf v^\circ \wedge V'({\mathbf v},H(\mathbf{v}),z) : ({\mathbf v},z) \leftarrow {\cal A}_+^{H} \bigr] \, .
	\end{align*}
	Summing both sides of the inequality over $h_i^\circ$ for $i\geq 2$ yields
	\begin{align}
	\begin{split}
	\Pr_{{{\mathbf \Theta}}}\bigr[{\mathbf x}\!=\! \mathbf x^\circ& \wedge V'({\mathbf v},{{\mathbf \Theta}},z) : (\pi,\pi({\mathbf v}),z) \leftarrow \langle{\hat{\cal S}}^{\cal A_+} , \pi({\mathbf \Theta})\rangle\bigl] 
	\\ &
	\geq \frac{1}{(q+n+1)^{2n}} \Pr_H\bigl[{\mathbf x}\!=\! \mathbf x^\circ \wedge V'({\mathbf v},H(\mathbf{v}),z) : ({\mathbf v},z) \leftarrow {\cal A}_+^{H} \bigr]
		 \\ &
		= \frac{1}{(q+n+1)^{2n}} \Pr_H\bigl[\mathbf{x}\!=\! \mathbf x^\circ \wedge V({\mathbf x},\mathbf{h}^{H,\mathbf{x}},z) : ({\mathbf x},z) \leftarrow {\cal A}^{H} \bigr] \, .\end{split}\label{eq:plainapplication}
	\end{align}
	Recalling its construction, the simulator ${\hat{\cal S}}^{\cal A_+}$ begins by sampling a uniformly random permutation $\pi$, so we can write
	\begin{align}
	\begin{split}
	\Pr_{{{\mathbf \Theta}}}\bigr[&{\mathbf x}\!=\! \mathbf x^\circ \wedge V'({\mathbf v},{{\mathbf \Theta}},z) : (\pi,\pi({\mathbf v}),z) \leftarrow \langle{\hat{\cal S}}^{\cal A_+} , \pi({\mathbf \Theta})\rangle\bigl]  
	\\ &
	=\frac 1 {n!}\sum_{\sigma\in S_n}\Pr_{{\mathbf \Theta}}\bigr[{\mathbf x}\!=\! \mathbf x^\circ \wedge V'({\mathbf v},{{\mathbf \Theta}},z): (\pi,\pi(\mathbf v),z) \leftarrow \langle{\hat{\cal S}}^{\cal A_+} , {\pi(\mathbf\Theta)}\rangle\big|\pi=\sigma\bigl]  \, .
	\end{split}
	\label{eq:permsum}
	\end{align}
	By definition, the predicate $V'({\mathbf v},\mathbf{\Theta},z) $ (with $\mathbf{v}$ of the form as explained above) is false whenever there exists an $i\geq 2$ such that $h_i\neq \Theta_{i-1}$. Now suppose that $\pi\neq\mathrm{id}$, then there must be some $j$ such that $\pi(j)<\pi(j-1)$. This implies that the first $\pi(j)$ stages of $\hat{\cal S}^{\cal A_+}$ which together (in the $\pi(j)$-th stage) produce $v_j=(h_j,x_j)$ are independent of $\Theta_{j-1}$, since $\Theta_{j-1}$ is given as input only at the {\em later} stage $\pi(j-1)$. We thus have the following, taking it as understood, here and in the sequel, that the random variables $\pi,\mathbf{v},\mathbf{\Theta}$ and $z$ are as in~(\ref{eq:permsum}). 
	\begin{align*}
	\Pr\bigl[{\mathbf x}\!=\! \mathbf x^\circ \wedge V'({\mathbf v},{{\mathbf \Theta}},z)\big|\pi\neq\mathrm{id}\bigr]
	&\le \Pr\bigl[{\mathbf x}\!=\! \mathbf x^\circ \wedge   h_j=\Theta_{j-1}|\pi\neq\mathrm{id}\bigr] 
	=\frac{\Pr\bigl[{\mathbf x}\!=\! \mathbf x^\circ|\pi\neq\mathrm{id}\bigr]}{|{\cal Y}|} \, .
	\end{align*}
	Using Equation \eqref{eq:permsum}, we can bound
	\begin{align*}
	\frac 1 {n!}\sum_{\sigma\in S_n}\Pr\bigr[{\mathbf x}\!=\! \mathbf x^\circ \wedge V'({\mathbf v},{{\mathbf \Theta}},z)\big|\pi\!=\!\sigma\bigl] 
	&\leq \frac 1 {n!}\Pr\bigr[{\mathbf x}\!=\! \mathbf x^\circ \wedge V'({\mathbf v},{{\mathbf \Theta}},z)\big|\pi\!=\!\mathrm{id}\bigl]
	+\frac{\Pr\bigl[{\mathbf x}\!=\! \mathbf x^\circ|\pi\!\neq\!\mathrm{id}\bigr]}{|{\cal Y}|} \, .
	\end{align*}
	We note that by definition of $V'$,
	\begin{align*}
	\Pr\bigr[{\mathbf x}\!=\! \mathbf x^\circ \wedge V({\mathbf x},{{\mathbf \Theta}},z)\big|\pi=\mathrm{id}\bigl] &\geq 
	\Pr\bigr[{\mathbf x}\!=\! \mathbf x^\circ \wedge V'({\mathbf v},{{\mathbf \Theta}},z)\big|\pi=\mathrm{id}\bigl]\ .
	\end{align*}
	Furthermore, we may define a new simulator ${\cal S}$ which takes oracle access to $\cal A$ and turns it into $\cal A_+$, and always chooses $\pi=\mathrm{id}$ instead of a random permutation. Where $\hat{\cal S}$ would output $(\mathbf v,z)$, ${\cal S}$ ignores the $\mathbf{h}$-part of $\mathbf{v}$ and simply outputs $(\mathbf x,z)$. We then have
	\begin{align*}
	\Pr_{{{\mathbf \Theta}}}\bigr[&\mathbf{x}\!=\! \mathbf{x}^\circ \wedge V({\mathbf x},{{\mathbf \Theta}},z) : ({\mathbf x},z) \leftarrow \langle{{\cal S}^A} , {\mathbf	 \Theta}\rangle\bigl]
	\\ &
	\geq \frac{n!}{(q+n+1)^{2n}} \Pr_H\bigl[\mathbf{x}\!=\! \mathbf x^\circ \wedge V({\mathbf x},\mathbf{h}^{H,\mathbf{x}},z) : ({\mathbf x},z) \leftarrow {\cal A}^{H} \bigr] -\epsilon_{\mathbf x^\circ}\, .
	\end{align*}
	with $\epsilon_{\mathbf x^\circ}$ given by $\epsilon_{\mathbf x^\circ}:= n!\cdot\Pr_{\mathbf \Theta}\bigl[{\mathbf x}= \mathbf x^\circ|\pi\neq\mathrm{id}\bigr]/|{\cal Y}|$. 
	\qed\end{proof}

\section{The multi-round Fiat-Shamir transformation}\label{sec:mFS}

A straightforward generalization of the Fiat-Shamir transformation can be applied to arbitrary (i.e., multi-round) public-coin interactive proof systems (PCIP). We show here security of this multi-round Fiat-Shamir transformation in the QROM. 

\subsection{Public coin interactive proofs and multi-round Fiat-Shamir}

We begin by defining PCIPs, mainly to fix notation, and the corresponding multi-round Fiat-Shamir transformation. 

	\begin{defi}[Public coin interactive proof system (PCIP)]
	A $(2n\!+\!1)$-round public coin interactive proof system (PCIP) $\mathsf{\Pi} = ({\cal P}, {\cal V})$ for a 
	language $\mathcal{L}$ is a $(2n\!+\!1)$-round two-party interactive protocol of the form, with $\cal C$ being a finite non-empty set, and $V$ a predicate:
	\begin{empheq}[box=\widefbox]{align*}
	&\underline{\text{Prover } {\cal P}(x)}&&&&\underline{\text{Verifier } {\cal V}(x)}\\
	&&&\overset{a_1}{\longrightarrow}&\\
	&&&\overset{c_1}{\longleftarrow}&&c_1\overset{\,\$}{\leftarrow} {\cal C} \\
	&&&\quad\vdots&&\\
	&&&\overset{a_n}{\longrightarrow}&\\
	&&&\overset{c_n}{\longleftarrow}&&c_n\overset{\,\$}{\leftarrow} {\cal C} \\
	&&&\overset{z}{\longrightarrow}&&\textup{Accept iff } V(x,a_1,c_1,...,a_n,c_n,z) = 1
	\end{empheq}\switch{\vspace{-1.5ex}}{}
%
\end{defi}

\begin{rem}
 If the language $\mathcal L$ is definied by means of an (efficiently verifiable) witness relation $R \subseteq {\cal X} \times {\cal W}$, then the prover typcially gets a witness $w$ for $x$ as an additional input. We then also say that $\mathsf \Pi$ is a PCIP {\em for the relation $R$}.
	In case of a $(2n\!+\!1)$-round PCIP  $\mathsf{\Pi}$ for a witness relation $R$ that is {\em hard on average}, meaning that there exists an instance generator $\Gen$ with the property that for $(w,x) \leftarrow \Gen$ it holds that $(w,x) \in R$, but given $x$ alone it is computationally hard to find $w$ with $(w,x) \in R$, $\mathsf{\Pi}$ is also called an {\em identification scheme}. 
\end{rem}


Just as in the ordinary Fiat-Shamir transformation, the interaction used to enforce the time order between the prover committing to the message $a_i$ and receiving the challenge $c_i$ can be replaced by means of a hash function. In addition, we can include the previous challenge (i.e. the previous hash value) in the hash determining the next challenge to enforce the ordering of the $n$ pairs $(a_i, c_i)$ according to increasing $i$. We thus obtain the following non-interactive proof system. 

\begin{defi}[Fiat-Shamir transformation for general PCIP (mFS)]\label{def:mFS}\ \\
	Given an $(2n\!+\!1)$-round PCIP $\mathsf{\Pi} = ({\cal P}, {\cal V})$ for a 
	language $\mathcal{L}
	$ 
	 and a hash function $H$ with appropriate domain, and range equal to $\cal C$, we define the non-interactive proof system $\mathsf{FS[\Pi]}= ({\cal P}^H_{FS}, {\cal V}^H_{FS})$  as follows. The prover $\cal P$ outputs
	\begin{align*}
		(x,a_1,...,a_n,z)&\leftarrow  {\cal P}^H_{FS}
	\end{align*}
	where $z$ and $a_i$ for $i=1,...,n$ are computed using $\cal P$, and the challenges are computed as
	\begin{align*}
		c_1&=H(0,x,a_1)\text{ and} \\
		c_i&=H(i-1,c_{i-1},a_i) \text{ for } i=2,...,n \, ,
	\end{align*}
	The verifier outputs `accept' iff  $V(x,a_1,c_1,...,a_n,c_n,z) = 1$ for $c_1=H(0,x,a_1)$ and  $c_i=H(i-1,c_{i-1},a_i)$, $i=2,...,n$, denoted by $V_{FS}(x,a_1,c_1,...,a_n,c_n,z) = 1$.
\end{defi}
\begin{rem}\label{rem:prefixes}
	The challenge number $i$ (minus 1) is included in the hash input to ensure that the challenges are generated using distinct inputs to $H$ with probability 1. This is to enable us to apply Theorem \ref{thm:enforced}, which only holds for duplicate-free lists of hash inputs. In fact, any additional strings can be included in the argument when computing $c_i$ using $H$, without influencing the security properties of the non-interactive proof system in a detrimental way. In the literature one sometimes sees that the entire previous transcript is hashed (in which case the counter number $i$ may then be omitted). 
\end{rem}

\subsection{General security of multi-round Fiat-Shamir in the QROM}

 When constructing a reduction for mFS, this reduction is participating as a prover in the underlying PCIP, and is hence only provided with random challenges one at a time. We thus need the special simulator from Theorem \ref{thm:enforced}, which always outputs the corresponding messages in the right order. The success of this simulator is based on the very essence of the Fiat-Shamir transformation, namely the fact that the intractability of the hash function takes the role of the interaction in enforcing a time order in the transcript of the PCIP.

The security of the multi-round Fiat-Shamir transformation follows as a simple Corollary of Theorem \ref{thm:enforced}.

\begin{cor}\label{cor:mFS} 
	There exists a black-box quantum polynomial-time $(n\!+\!1)$-stage quantum algorithm $\cal S$ such that for any adaptive adversary $\cal A$ against the multi-round Fiat-Shamir transformed version $\mathsf{FS[\Pi]}$ of a $(2n\!+\!1)$-round PCIP $\mathsf\Pi$, making $q$ queries to a uniformly random function $H$ with appropriate domain and range equal $\cal C$, and for any $x^\circ\in {\cal X}$:
	\begin{align*}
	\Pr\bigr[&x = x^\circ\wedge v = accept :(x,v) \leftarrow \langle{\cal S}^{\cal A} , {\cal V}\rangle\bigl] 
	\switch{\;}{\\ &}
	\geq \frac{n!}{(2q+n+1)^{2n}} \Pr_H\bigr[x = x^\circ \wedge V^H_{FS}(x,\pi) : (x,\pi) \leftarrow {\cal A}^H \bigl] -\epsilon_{x^\circ}\, .
	\end{align*}
	where the additive error term $\epsilon_{x^\circ}$ is equal to $\frac{n!}{|{\cal C}|}$ when summed over all $x^\circ$.
\end{cor}
\begin{proof}
	We may simply set $\mathbf{x^\circ} = (x^\circ,(0,a_1),\ldots,(n-1,a_n))$ for arbitrary $a_1,\ldots,a_n$, apply Theorem \ref{thm:enforced} and then sum over all choices of $a_1,\ldots,a_n$ to obtain the claimed inequality. Note that the round indices ensure that every such $\mathbf{x^\circ}$ is duplicate free, satisfying the corresponding requirement of Theorem \ref{thm:enforced}.
\end{proof}

Note that the additive error terms reflect the fact that the random oracle only \emph{approximately} succeeds in enforcing the original time order in the transcript of the PCIP. However, it can be made arbitrarily small, as discussed below.
\begin{rem}\label{rem:extendC}
	There exist PCIPs with soundness error much smaller than $1/|\mathcal C|$. As an example, consider the sequential repetition of a \sigp with special soundness. Here, the soundness error is $1/|\mathcal C|^n$. In this case, the term proportional to $1/|\mathcal C|$ renders the bound from the above theorem trivial. Note however, that (i) this situation is extremely artificial, as there is absolutely no reason to repeat sequentially instead of in parallel, and (ii) the additive error term can be made arbitrarily small by considering a variant $\mathsf\Pi'$ of $\mathsf\Pi$ where the random challenges are enlarged with a certain number of bits that are ignored otherwise, see Remark \ref{rem:decrease-additive-error}. 
	
	In fact, we suspect that the observation from (i) is true in a much broader sense: if a PCIP still has negligible soundness error when allowing the adversary to learn one of the challenges $c_i$ in advance of sending the corresponding commitment-type message $a_i$, it seems like the number of rounds can be reduced and the loss in soundness error can be won back by parallel repetition.
\end{rem}

As for the case of the Fiat-Shamir transformation for \sigps, the general reduction implies that security properties that protect against dishonest provers carry over from the interactive to the non-interactive proof system. For a definition of the properties considered in the following theorem, see, e.g. \cite{DFMS19}. The quantum proof-of-knowledge-property was intoduced in \cite{Unruh2012}.
\begin{cor}[Preservation of Soundness/PoK]\label{cor:PresSoundPoK}
	Let $\mathsf{\Pi}$ be a constant-round PCIP  that has (statistical/computational) soundness, and/or the (statistical/computational) quantum proof-of-knowledge-property, respectively. 
Then, in the QROM, $\mathsf{FS[\Pi]}$ has (statistical/computational) soundness, and/or the (statistical/computational) quantum proof-of-knowledge-property, too.
\end{cor}
\begin{proof}
	Corollary \ref{cor:mFS} turns any dishonest prover ${\cal A}_{\mathsf{FS[\Pi]}}$ for $\mathsf{FS[\Pi]}$ with success probability $\epsilon$ into a dishonest prover ${\cal A}_{\mathsf{\Pi}}$ for $\mathsf{\Pi}$, with success probability $\epsilon\cdot(2q+1)^{-2n}$, where $2n+1$ is the number of rounds in $\mathsf{\Pi}$. Since $n$ is constant and $q$ is polynomial in the security parameter, the success probabilities of the respective provers are polynomially related. The claimed implications follow now using the same arguments as in Corollaries 13 and 16 in \cite{DFMS19}.
\qed\end{proof}

\section{Tightness of the reductions}

Here, we show tightness of our results. We start with proving tightness of Theorems \ref{thm:main} and \ref{thm:FS} (up to essentially a factor $4$). This implies that a $O(q^2)$-loss is unavoidable in general. Indeed, the following result shows that for a large and natural class of \sigps $\mathsf{\Sigma}$, there exists an attack against $\mathsf{FS[\Sigma]}$ that succeeds with a probability $q^2$ times larger than the best attack against $\mathsf{\Sigma}$. The attack is based on an application of Grover's quantum algorithm for unstructured search. 

To our surprise, we could not find an analysis of Grover's algorithm in the regime we require in the literature. Grover search has been analyzed in the case of an unknown number of solutions \cite{BBGT98}, but the focus of that work is on analyzing the expected number of queries required to find a solution, while we analyze the probability with which the Grover search algorithm succeeds for  a \emph{fixed but arbitrary} number of queries. 

\begin{thm}\label{thm:qsqauredboost}
	Let ${\cal L}$ be a language, and let $\mathsf{\Sigma}$ be a \sigp for ${\cal L}$ with challenge set ${\cal C}$, special soundness and perfect honest-verifier zero-knowledge. Furthermore, we assume that the triples $(a,c,z)$ produced by the simulator ${\cal S}_{\mathrm{ZK}}(x)$ are always accepted by the verifier even for instances $x \not\in  \cal L$, and that $a$ has min-entropy $\gamma$.%
	\footnote{These additional assumptions on the simulator could be avoided, but they simplify the proof. Furthermore, for typical \sigps they are satisfied. In particular, the simulated transcripts for hard instances are accepted by the verifier with high probability. Otherwise, the two polynomial-time algorithms could otherwise be used to solve the hard instances, a contradiction. } Then for any $q$ such that $(q^2+1)\cdot e^2\cdot (5q)^6 < |{\cal C}|$ and $2^\gamma/(5q)^3 > 2$, there exists a $q$-query dishonest prover that succeeds with probability $q^2/|{\cal C}|$ in producing a valid $\mathsf{FS[\Sigma]}$-proof for an instance $x \not\in \cal L$. 
\end{thm}
The idea of the attack against $\mathsf{FS[\Sigma]}$ is quite simple. For a \sigp that is {\em special} honest-verifier zero-knowledge, meaning that the simulation works by first sampling the challenge $c$ and the repsonse $z$ and then computing a fitting answer $a$ as a function $a(c,z)$, one simply does a Grover search to find a pair $(c,z)$ for which $H\bigl(x,a(c,z)\bigr) = c$. For a typical $H$, this will give a quadratic improvement over the classical search, which, for a random $H$, succeeds with probability $q/|{\cal C}|$ (due to the special soundness). A subtle issue is that, for some (unlikely) choices of $H$, there are actually {\em many} $(c,z)$ for which $H\bigl(x,a(c,z)\bigr) = c$, in which case the Grover search ``overshoots''. In the formal proof below, this is dealt with by controlling the probability of $H$ having this (unlikely) property. Also, it removes the {\em special} honest-verifier zero-knowledge property by doing the Grover search over the randomness of the simulator, which requires some additional caution. 

\begin{rem}\label{rem:AttackExt}
It is not hard to see that Theorem~\ref{thm:qsqauredboost} still holds in the following two variations of the statement. (1) $H(x,a)$ is random and independent for different choices of $a$, but is {\em not} necessarily independent for different choices of $x$. (2) The \sigp $\mathsf\Sigma$ is replaced by ${\mathsf\Sigma}'$, which has its challenge enlarged with a certain number of bits that are ignored otherwise, in line with Remark~\ref{rem:extendC}, and $\mathsf{FS[\Sigma']}$ then uses an $H$ with a correspondingly enlarged range.%
\footnote{While (1) follows by inspecting the proof, (2) holds more generically: the dishonest prover attacking $\mathsf{FS[\Sigma']}$ simply runs the prover attacking $\mathsf{FS[\Sigma]}$ but enlarges the output register of the hash queries, with the corresponding state being set to be the fully mixed state in each query, and then dismisses these additional qubits again.} 
\end{rem}

\begin{proof}
	Let ${\cal S}_{\mathrm{ZK}}$ be the zero-knowledge simulator given by the perfect honest-verifier zero-knowledge property of $\mathsf{\Sigma}$. Consider an adversary $\mathcal{A}_{FS}$ against $\mathsf{FS[\Sigma]}$, that works as follows for an arbitrary instance $x\notin \mathcal{L}$:
	\begin{itemize}
		\item Define the function $f^H: R\rightarrow \{0,1\}$ (where $R$ is the set of random coins for ${\cal S}_{\mathrm{ZK}}$) as 
		$$
		f^H(\rho) = \begin{cases}
		1&\text{for }{\cal S}_{\mathrm{ZK}}(x;\rho)\rightarrow (a,c,z) \wedge H(x||a) = c
		\\0&\text{otherwise}.
		\end{cases}
		$$
		\item Use Grover's algorithm for $q$ steps, to try and find $\rho$ s.t. $f(\rho) = 1$
		\item Run ${\cal S}_{\mathrm{ZK}}(x;\rho) \rightarrow (a,c,z)$ and output $(x,a||z)$. 
	\end{itemize}
	Let $p_1^H$ be the fraction of random coins from $R$ that map to $1$ under $f^H$. Note that by the special soundness of $\Sigma$, in any accepting triple $a$ determines $c$ and we thus have $\E_H[p_1^H] = \frac{1}{|\cal C|}$. By the way Grover works, after $q$ iterations (requiring $q$ queries to $H$) the probability $p_2^H$ of finding such an input is $\sin^2((2q+1)\Theta^H)$, where $0\leq \Theta^H \leq \pi/2$ is such that $\sin^2(\Theta^H) = p_1^H$. Now as long as $\Theta$ is not too large to begin with (i.e. as long as the Grover search will not `overshoot'), $p_2^H$ is approximately a factor $q^2$ larger than $p_1^H$. Our goal will be to show that also on average over $H$, the improvement is at least $q^2$. To this end we define $H_{\text{bad}} := \{H : p_1^H > \sin^2(\frac{\pi}{6q+3})\}$ and $H_{\text{good}}$ its complement. Then,
\begin{align*}
	\E_H[p_2^H] &= (1-\alpha)\cdot\E_{H}\left[p_2^H | H\in H_\text{good}\right] + \alpha\cdot \E_{H}\left[p_2^H|H\in  H_{\text{bad}}\right]\\
	&\geq (1-\alpha)\cdot\E_{H}\left[p_2^H|H\in H_\text{good}\right]
\end{align*}
	where $\alpha=\Pr_H[H\in H_{\text{bad}}]$ and $1-\alpha = \Pr_H[H\in H_{\text{good}}]$.
	
	We first compute $\E_{H_{\text{good}}}\left[p_2^H\right]$. Let $H\in H_{\text{good}}$. We have $(2q+1)\Theta^H \leq \frac{\pi}{3}$. Since $\frac{\text{d}}{\text{d}\Theta}\sin(\Theta) = \cos(\Theta)\geq 1/2$ for $\Theta\in [0,\frac{\pi}{3}]$, and $\Theta \geq \sin(\Theta)$, it follows that
	$$
	\sin((2q+1)\cdot\Theta^H) \qquad\geq\qquad \sin(\Theta^H) + \frac{2q\cdot\Theta^H}{2}\qquad \geq\qquad (q+1)\cdot\sin(\Theta^H).
	$$
	Using $\sin(\Theta)\geq 0$ for $\Theta\in [0,\frac{\pi}{3}]$, we obtain
	$$
	p_2^H = \sin^2((2q+1)\cdot\Theta^H) \geq (q+1)^2\cdot\sin^2(\Theta^H) = (q+1)^2\cdot p_1^H.
	$$
	Therefore,
	\begin{align}\label{eqn:p2}
	\begin{split}
	\E_H[p_2^H] \qquad & \geq \qquad\E_{H}\left[p_2^H | H\in H_\text{good}\right]\cdot \Pr_H[H\in H_{\text{good}}]\\
	&\geq\qquad(q+1)^2\cdot\E_{H}\left[p_1^H | H\in H_\text{good}\right]\cdot \Pr_H[H\in H_{\text{good}}]\\
	&\geq \qquad(q+1)^2\cdot\left(\E_H[p_1^H] - \Pr_H[H\in H_{\text{bad}}] \right). 
	\end{split}
	\end{align}
	
	Next we bound $\alpha = \Pr_H[H\in H_{\text{bad}}] = \Pr_H[p_1^H > \sin^2(\frac{\pi}{6q+3})]$. Note that for $p_1^H$ to be large, we need that for many first messages $a$, $H(a)$ must be the unique challenge $c$ for which there exist an accepting response. For a random $H$ this is unlikely to happen. Formally, we argue as follows, using the Chernoff bound eventually. 
	
	We first define the following equivalence relation:
	$$
	\rho \sim \rho' \text{ iff } {\cal S}_{\mathrm{ZK}}(\rho) = (a,c,z) \wedge {\cal S}_{\mathrm{ZK}}(\rho') = (a,c',z') \text{ for }\rho,\rho'\in R.
	$$
	$R/_{\!\sim}$ then denotes the set of equivalence classes $[\rho] = \{\rho' \in R \,|\,\rho \sim \rho'\}$. 
	By the perfect special soundness property and the assumptions on ${\cal S}_{\mathrm{ZK}}$, we have that $a$ determines $c$ (remember that $x\notin {\cal L}$), and therefore $f^H$ is constant on elements within a given equivalence class. 
	Thus, $f^H: R/_{\!\sim} \rightarrow \{0,1\}$. 
	For two distinct equivalence classes $[\rho]\neq [\rho']$, we have
	$$
	\Pr_H[f^H([\rho]) = 1 \wedge f^H([\rho']) = 1] = \Pr_H[f^H([\rho]) = 1]\cdot \Pr_H[f^H([\rho']) = 1] \, ,
	$$ 
	since $H(x||a)$ is chosen independently for different $a$. Finally, taking $X^H := \sum_{[\rho]} f^H([\rho])$ we have
	\begin{align*}
	p_1^H& = \Pr_\rho[f^H(\rho)=1] = \frac{\sum_{\rho} f(\rho)}{|R|}\\
	& = \frac{\sum_{[\rho]} \left(f^H([\rho])\cdot |[\rho]|\right)}{|R|} \leq \frac{|[\rho_{\max}]|\cdot\sum_{[\rho]} f^H([\rho])}{|R|} = X^H \cdot 2^{-\gamma}
	\end{align*}
	where $[\rho_{\max}]$ is the $[\rho]$ that maximizes $|[\rho]|$. It follows that 
	\begin{align*}
	\alpha& = \Pr_H[p_1^H > \sin^2\left(\frac{\pi}{6q+3}\right)] \\
	&\leq\Pr_H\left[X^H > \sin^2\left(\frac{\pi}{6q+3}\right)\cdot 2^\gamma\right]\leq \Pr_H\left[X^H > \frac{2^\gamma}{|{\cal C}|} + \frac{2^\gamma}{(5q)^3}\right]
	\end{align*}
	where we used $\sin^2(x)> x^3$ for $0\leq x \leq 0.80$ and $\frac{\pi}{6q+3} > \frac{1}{5q} + \sqrt[3]{\frac{1}{|{\cal C}|}}$ for ${|\cal C|} > (5q)^3$ in the last inequality.
	By definition of $f$, for any $[\rho]$ we have $\Pr_H\left[f(\rho)=1\right]=\frac{1}{|{\cal C}|}$, hence 
	\begin{align*}
	\E_{H}\left[X\right]= \sum_{[\rho]}\E_H[f^H([\rho])]=\sum_{[\rho]}\Pr_H[f^H([\rho]) =1] =\frac{|R/_{\!\sim}|}{|{\cal C}|}\geq \frac{2^\gamma}{|{\cal C}|}.
\end{align*}
	We use the following Chernoff bound:
	\begin{align*}
	\Pr_H\left[X^H > (1+\delta)\cdot \E
	_{H}\left[X^H\right] \right] &< \left(\frac{e^{\delta}}{(1+\delta)^{1+\delta}}\right)^{\E
		_{H}\left[X^H\right]} < \left(\frac{e^{1+\delta}}{\delta^{1+\delta}}\right)^{\E
		_{H}\left[X^H\right]}\\
	& = \left(\frac{e}{\delta}\right)^{\E
		_{H}\left[X^H\right]\cdot(1+\delta)}.
	\end{align*}
	Setting $\delta:=\frac{|{\cal C}|}{(5q)^3}$, together with the inequalities derived above this leads to		
	\begin{align*}
	\alpha \leq 
	\left(\frac{e\cdot (5q)^3}{|{\cal C}|}\right)^{\frac{2^\gamma}{|{\cal C}|} + \frac{2^\gamma}{(5q)^3}} 
	< \frac{e^2\cdot (5q)^6}{|{\cal C}|^2}< \frac{1}{|{\cal C}|\cdot (q^2+1)}
	\end{align*}
	where we used $\frac{2^\gamma}{(5q)^3} > 2$ in the second to last, and $|{\cal C}| > (q^2+1)\cdot e^2\cdot (5q)^6$ in the last inequality.
	Plugging this bound into Equation \ref{eqn:p2}, we get
	$$
	\E_{H}[p_2^H] \geq(q^2+1)\cdot \left(p_1-\frac{1}{|{\cal C}|\cdot(q^2+1)}\right) = \frac{q^2}{|{\cal C}|} + \frac{1}{|{\cal C}|} - \frac{1}{|{\cal C}|} = \frac{q^2}{|{\cal C}|}.
	$$
	Thus, the success probability of our adversary $\mathcal{A}_{FS}$ after making $q$ queries to $H$ is at least $\frac{q^2}{|{\cal C}|}$. 
\qed\end{proof}

%
%
%

The tightness of Corollary \ref{cor:mFS} follows from the above tightness result for the case of \sigps in a fairly straightforward manner.

\begin{thm}\label{thm:qtothe2nboost}
	For every positive integer $n$, there exists a $(2n\!+\!1)$-round PCIP $\mathsf{\Pi}$ with soundness error $\epsilon$  and challenge space $\mathcal C$ such that $|{\cal C}| \geq 1/\epsilon$ and such that there exists a $q$-query dishonest prover $\cal A$ on $\mathsf{FS(\Pi)}$ with success probability $n^{-2n}q^{2n}\epsilon$. 
\end{thm}
Before proving the theorem, we show how it implies the tightness of  Theorem  \ref{cor:mFS}.

\begin{cor}
	The security loss in the bound in Corollary \ref{cor:mFS} is optimal, up to a multiplicative factor that depends on $n$ only. 
\end{cor}
\begin{proof}
	Let  ${\mathsf\Pi}$ be a PCIP  as shown to exist in Theorem \ref{thm:qtothe2nboost}. Let $\epsilon_\Pi$, and $\epsilon_{\mathsf{FS(\Pi)}}(q)$, be the soundness error of $\mathsf \Pi$, and the one of  its Fiat Shamir transformation against $q$-query adversaries, respectively. By Theorem \ref{thm:qtothe2nboost}, 
		\begin{equation}
			\epsilon_{\mathsf{FS(\Pi)}}(q)\ge n^{-2n}q^{2n}\epsilon_{\mathsf \Pi}.
		\end{equation}
		Theorem \ref{cor:mFS}, on the other hand, yields
		\begin{align}
			\epsilon_{\mathsf \Pi}&\ge \frac{n!}{(2q+n+1)^{2n}}\epsilon_{\mathsf{FS(\Pi)}}(q)-\frac{n!}{|\mathcal C|}\\
				&\ge \frac{n!}{(2q+n+1)^{2n}}\epsilon_{\mathsf{FS(\Pi)}}(q)-n!\epsilon_{\mathsf \Pi},
		\end{align}
		where we used the condition on the challenge space size from Theorem \ref{thm:qtothe2nboost} in the last line. Rearranging terms we obtain
		\begin{align}
		\epsilon_{\mathsf{FS(\Pi)}}(q)&\le (2q+n+1)^{2n}\left(1+\frac 1{n!}\right)\epsilon_{\mathsf{\Pi}}(q)\\
		&\le 2(n+3)^2q^{2n}\epsilon_{\mathsf{\Pi}}(q),
			\end{align}
			where we have used $1\le q$ in the last line. In summary, we have constants $c_1=n^{-2n}$ and  $c_2= 2(n+3)^{2n}$ such that
			\begin{equation}
				c_1 q^{2n}\epsilon_{\mathsf \Pi}\le \epsilon_{\mathsf{FS(\Pi)}}(q)\le c_2 q^{2n}\epsilon_{\mathsf \Pi}.
			\end{equation}
\qed\end{proof}

\begin{proof} [of Theorem \ref{thm:qtothe2nboost}]
	Let $\hat{\mathsf \Sigma}$ be a \sigp for a language $\cal L$ fulfilling the requirements of Theorem~\ref{thm:qsqauredboost}. Let the challenge space be denoted by $\hat{ \mathcal C}$. Given an arbitrary positive integer, we define an $(2n\!+\!1)$-round PCIP $\mathsf\Pi$ for the same language $\cal L$ by means of $n$ sequential independent executions of $\hat{\mathsf \Sigma}$ . Concretely, the $2n+1$ messages of $\mathsf \Pi$ are given in terms of the messages  $\hat a_i, \hat c_i$ and $\hat z_i$  of the $i$-th repetition of $\hat{\mathsf \Sigma}$ as 
	\begin{align*}
		a_1&=\hat a_1\\
		c_i&=(\hat c_i, r_i)\  \mathrm{for}\ i=1,...,n\\
		a_i&=(\hat a_i, \hat z_{i-1})\ \mathrm{for}\ i=2,...,n, \ \mathrm{and}\\
		z&=\hat z_{n},
	\end{align*}
	where $r_i$ is an independent random string of arbitrary (but fixed) length, which is ignored otherwise (in line with Remark~\ref{rem:extendC}). 
	The purpose of $r_i$ is to make the challenge space $\cal C$ of $\mathsf\Pi$ arbitrary large, as required. 
	The verification procedure of $\mathsf \Pi$ simply checks if all the triples $(\hat a_i, \hat c_i, \hat z_i)$ are accepted by $\hat{\mathsf \Sigma}$. 
	By the special soundness property of $\hat{\mathsf \Sigma}$, the soundness error of this PCIP is $\epsilon=|\hat{\cal C}|^{-n}$. 
	
	Using Theorem \ref{thm:qsqauredboost}, we can attack the Fiat-Shamir transformation of $\hat{\mathsf \Sigma}$ repeatedly to devise an attack agains $\mathsf{FS(\Pi)}$: first use Theorem~\ref{thm:qsqauredboost} to find $\hat a_1$ and $\hat z_1$, then use it again to find $\hat a_2$ and $\hat z_2$, etc., having the property that with the correctly computed challenges these form valid triples for an instance $x \not\in \cal L$. In each invocation of Theorem~\ref{thm:qsqauredboost} we use a $q'$-query attack, which then succeeds with probability $q'^2/|{\cal \hat{C}}|$. Thus, using in total $q = n q'$ queries, we succeed in breaking $\mathsf{FS[\Pi]}$ with probability $q'^{2n}/|{\cal \hat{C}}|^n = n^{-2n}q^{2n}\epsilon$, as claimed. 
	
	There are two issues we neglected in the above argument. First, we actually employ Theorem~\ref{thm:qsqauredboost} for attacking a {\em variant} of $\hat{\mathsf \Sigma}$ that has its challenge enlarged (and thus is not special sound); and, second, the challenge $c_i$ is computed as 
	$$
	c_i = H(i-1,...,H(1,H(0,x,\hat a_1),\hat a_2),...,\hat a_i) \, ,
	$$
	which is {\em not} a uniformly random function of $x$ and $\hat a_i$ (but only of $\hat a_i$). However, by Remark~\ref{rem:AttackExt}, the attack from Theorem \ref{thm:qsqauredboost} still applies. 
\qed\end{proof}

\section{Applications}
\subsection{Digital signature schemes from multi-round Fiat-Shamir}\label{sec:sig}

One of the prime applications of the Fiat-Shamir transformation is the construction of  digital signature schemes from interactive identification schemes. In this context, multi-round variants have also been used. An example where a QROM reduction is especially desirable is MQDSS \cite{MQDSS}, a candidate digital signature scheme in the ongoing NIST standardization process for post-quantum cryptographic schemes \cite{NIST}. This digital signature scheme is constructed by applying the multi-round Fiat-Shamir transformation to the 5-round identification scheme by  Sakumoto, Shirai, and Hiwatari \cite{SSH} based on the hardness of solving systems of multivariate quadratic equations.

 In this section, we present a generic construction of a digital signature scheme based on multi-round FS, and give a proof sketch of its strong unforgeability under chosen message attacks. We refrain from giving a full, self-contained proof here so as to not distract from our main technical result and its implications. Many, though not all, parts of the argument are very similar to the ones made elsewhere for the 3-round case.

The following construction is a straightforward generalization of the original construction of Fiat and Shamir.

\begin{defi}[Fiat-Shamir signatures from a general PCIP]\label{def:mFS-sig}
	Given an $(2n\!+\!1)$-round public coin identification scheme $\mathsf{\Pi} = ({\Gen},{\cal P}, {\cal V})$ for a witness 
	relation $R$ 
	and a 
	hash function $H$ with appropriate domain and range equal to $\cal C$, we define the digital signature scheme $\mathsf{Sig[\Pi]}= (\Gen, \Sign, \Ver)$  as follows. The key generation algorithm $\Gen$ is just the one from $\Pi$. The signing algorithm $\Sign$, on input a secret key $sk$ and a message $m$, outputs
	\begin{align*}
	\sigma=(a_1,...,a_n,z)&\leftarrow  \Sign_{sk}(m)
	\end{align*}
	where $z$ and $a_i$ for $i=1,...,n$ are computed using ${\cal P}(pk)$, and the challenges are computed as 
	\begin{align*}
	c_1&=H(0, pk,m,a_1)\text{ and} \\
	c_i&=H(i-1,c_{i-1},a_i) \text{ for } i=2,...,n \, .
	\end{align*}
	The verification algorithm $\Ver$, on input a public key $pk$, a message $m$ and a signature $\sigma=(a_1,...,a_n,z)$, computes $c_i$ as specified above, outputs `accept' iff  ${\cal V}_{pk}(a_1,c_1,...,a_n,c_n,z) = 1$, denoted by $\Ver_{pk}(m,\sigma) = 1$.
	\\~\\
	We note that the above definition is equivalent to the following, alternative formulation: Let\linebreak $\Sign_{sk}(m)$ produce $\sigma$ by running $P_{FS}^H(x||m)$, and let $\mathsf{Verify}(m,\sigma)$ be equal to the outcome of\linebreak $V_{FS}^H(x||m)$, where $(P_{FS}^H,V_{FS}^H) = \mathsf{FS[\Pi^*]}$ and $\mathsf{\Pi^*} = ({\cal P^*,\cal V^*})$ is the identification scheme obtained from $\mathsf{\Pi}$ by setting ${\cal P^*}(x||m) = {\cal P}(x)$ and ${\cal V^*}(x||m) = {\cal V}(x)$ for any $m$. This alternative formulation will be convenient in the proof of Theorem \ref{thm:(s)UFCMA}.
\end{defi}
\begin{rem}
		As in the case of the plain multi-round Fiat-Shamir transformation, one can include arbitrary additional strings in the argument when computing the challenges $c_i$. Examples where this is done include the MQDSS signature scheme \cite{MQDSS}, where the message $m$ and the first commitment $a_1$ are also included in the argument for computing the second challenge, and Bulletproofs, where the challenges are computed by hashing the entire transcript up to that point \cite{Bulletproofs}.
\end{rem}
As an identification scheme is an interactive honest-verifier zero knowledge proof of knowledge of a secret key, the above signature scheme is a a non-interactive zero knowledge proof of  knowledge of a secret key according to Corollary \ref{cor:mFS}. For a digital signature scheme, however, the stronger security notion of (strong) unforgeability against chosen message ((s)UF-CMA) attacks is required.

In the following, we give a proof sketch for the fact that the above signature scheme is (s)UF-CMA. This fact follows immediately once we have convinced ourselves that a certain result by Unruh about the Fiat-Shamir transformation holds for the multi-round case as well: For the Fiat-Shamir transformation of \sigps, extractability implies a stronger notion of extractability enabling a proof of (s)UF-CMA \cite{Unruh2017}. Here, we just patch the parts of the proof from \cite{Unruh2017} that make use of the fact that the underlying PCIP has only three rounds.

For the following we need the notion of a PCIP having computationally unique responses.
\begin{defi}[Computationally unique responses - PCIP]
A $(2n\!+\!1)$-\switch{\linebreak}{}round PCIP 
$\mathsf{\Pi} = (\cal P, \cal V)$ is said to have {\em computationally unique responses} if given a partial transcript\switch{}{\linebreak} $(x,a_1,c_1,\ldots a_i,c_i)$ it is computationally hard to find two accepting conversations that both extend the partial transcript but differ in (at least) $a_{i+1}$ (here we consider $z$ to be equal to $a_{n+1}$), i.e.~for  $con_i=x,a_1,c_1,\ldots a_i,c_i,a_{i+1}^{(j)},c_{i+1}^{(j)}\ldots,a^{(j)}_n,c^{(j)}_n,z^{(j)}$, $j=1,2$ we have that 
$$
\Pr\left[{{\cal V}(con_1) = 1 \wedge {\cal V}(con_2) = 1}: (con_1,con_2)\leftarrow {\cal A}\right]
$$
is negligible for computationally bounded (quantum) ${\cal A}$, where $a^{(1)}_{i+1}\neq a^{(2)}_{i+1}$.

\end{defi}
Equipped with this definition, we can state the main result of this section.
\begin{thm}[(s)UF-CMA of multi-round FS signatures]\label{thm:(s)UFCMA}
	Let $\mathsf{\Pi}$ be a \switch{\linebreak}{}PCIP for some hard relation $R$, which is a quantum proof of knowledge and satisfies completeness, HVZK, and has unpredictable commitments\footnote{We take unpredictable commitments for PCIP's to be exactly the same as for \sigps, with the first message playing the role of the commitment.} as well as a superpolynomially large challenge space. Then $\mathsf{Sig[\Pi]}$ is existentially unforgeable under chosen message attack (UF-CMA).
If  $\mathsf{\Pi}$ in addition has computationally unique responses, $\mathsf{Sig[\Pi]}$ is {\em strongly} existentially unforgeable under chosen message attack (sUF-CMA).
\end{thm}
In \cite{Unruh2017} (Theorem 24, and  25,  respectively), it is proven that an extractable FS proof system (of an HVZK \sigp, and of an HVZK \sigp with computationally unique responses, respectively) satisfies the stronger notion of {\em (strong) simulation-sound extractability}. In addition, it is shown that such a FS proof system gives rise to a (s)UF-CMA signature scheme if the underlying relation is hard. Corollary \ref{cor:PresSoundPoK} implies that $\mathsf{FS[\Pi^*]}$ is indeed extractable if $\mathsf\Pi$ is extractable. Below we rely on the proof in \cite{Unruh2017} to argue simulation-sound extractability, only pointing out a particular difference for the multi-round case.
\begin{proof}[sketch]
Since $\mathsf{\Pi}$ is a quantum proof of knowledge, so is $\mathsf{\Pi^*}$. By Corollary \ref{cor:PresSoundPoK}, $\mathsf{FS[\Pi^*]}$ is a quantum proof of knowledge (extractable), and by Theorem 20 in \cite{Unruh2017} (which easily generalizes to the multi-round setting), completeness, unpredictable commitments\footnote{This property is required to have sufficient entropy on the inputs to the oracle that are reprogrammed by the zero-knowledge simulator ${\cal S}_{ZK}$. While ${\cal S}_{ZK}$ may reprogram the oracle on inputs $(i-1,c_{i-1},a_i)$ for $i>1$, it is enough to require the first message $a_1$ to have sufficient entropy, since with $c_{i-1}$, these later inputs all include a uniformly random element from the superpolynomially large challenge space.} and HVZK of $\mathsf{\Pi^*}$ together imply ZK for $\mathsf{FS[\Pi^*]}$. For the proof that $\mathsf{FS[\Pi^*]}$ is also simulation-sound extractable, we refer to the proof of Theorem 24 in \cite{Unruh2017}, noting only that in the hop from Game 1 to Game 2 we have to adjust the argument as follows: Let ${\cal S}_{ZK}$ be the zero-knowledge simulator that runs the HVZK simulator from $\mathsf{\Pi}^*$ and reprograms the oracle as necessary. We write $H_f$ for the oracle $H$ after it has been reprogrammed by ${\cal S}_{ZK}$, at the end of the run of ${\cal A}$. We have to show that $V_{FS}^{H_f}(x,a_1,\ldots,a_n,z) = 1$ implies $V_{FS}^{H}(x,a_1,\ldots,a_n,z) = 1$, where $(x,a_1,\ldots,a_n,z)$ is the final output of ${\cal A}$. Suppose the implication does not hold. Then either (i) $\ H_f(0,x,a_1)\neq H(0,x,a_1)$ or (ii) $ H_f(i-1,c_{i-1},a_i)\neq H(i-1,c'_{i-1},a_i)$ for some $i$, where $c_{i-1}$ is the $(i\!-\!1)$-st challenge as recomputed by $V_{FS}^{H_f}$ and $c'_{i-1}$ is the one computed by $V_{FS}^H$. In case (i) holds, ${\cal A}$ has queried $x$ and the corresponding forged proof that was output by ${\cal S}_{ZK}$ starts with $a_1$. In case (ii), assume that $H_f(j-1,c_{j-1},a_j)= H(j-1,c_{j-1},a_j)$ for all $j < i$, so that $c_{i-1} = c'_{i-1}$. Then,
$$
H_f(i-1,...,H(1,H(0,x, a_1), a_2),...,a_i) \neq H(i-1,...,H(1,H(0,x, a_1), a_2),...,a_i)
$$
which means that ${\cal A}$ either queried $x$ and the corresponding forged proof that was output by ${\cal S}_{ZK}$ starts with $a_1$, or else ${\cal A}$ has queried some $x'$ such that
\begin{align*}
H(i-2,\ldots,H(&1,H(0,x', a'_1),a'_2),\ldots a'_{i-1}) \\
&= H(i-2,\ldots,H(1,H(0,x, a_1),a_2),\ldots, a_{i-1})
\end{align*}
and $a_i = a'_i$, where $(a'_1,\ldots,a'_i)$ is part of the ${\cal S}_{ZK}$ proof resulting from the query $x'$. By the fact that $H$ is a random oracle, it is infeasible for $\mathcal{A}$ to find such an $x'$. 

In the context of weak simulation-sound extractability, the fact that ${\cal A}$ has queried $x$ is enough to derive a contradiction. For the strong variant, we now have that ${\cal S}_{ZK}$ has output $(x,a_1,a'_2,\ldots,a'_n,$\switch{}{\linebreak} $z')$ such that $${\cal V}(x,a_1,H_f(0,x,a_1),a'_2,c'_2\ldots,a'_n,c'_n,z') = 1$$ and ${\cal A}$ has output $(x,a_1,a_2,\ldots,a_n,z)$ such that $${\cal V}(x,a_1,H_f(0,x,a_1), a_2,c_2,\ldots,a_n,c_n,z) =1$$ (and ${\cal A}$ knows both since it interacted with ${\cal S}_{ZK}$). By the computationally unique responses property of $\mathsf{\Pi}$, it must be that $a_2 = a_2'$. But then it follows that $$c_2 = H_f(1,H_f(0,x,a_1), a_2) = H_f(1,H_f(0,x,a_1), a'_2) = c'_2$$ (remember that both proofs are accepting with respect to $H_f$) which in turn implies that $a_3 = a'_3$, etc. Thus, we obtain that ${\cal A}$ has output a proof that was produced by ${\cal S}_{ZK}$, yielding a contradiction. We conclude that $$V_{FS}^{H_f}(x,a_1,\ldots,a_n,z) = 1\text{ implies }V_{FS}^{H}(x,a_1,\ldots,a_n,z) = 1$$ except with negligible probability.

In the rest of the proof of Theorems 24 and 25 in \cite{Unruh2017}, no properties specific to a three-round scheme are used, and so the results extend to the PCIP context, that is, $\mathsf{FS[\Pi^*]}$ is (strongly) simulation-sound extractable. Now applying Theorem 31 from \cite{Unruh2017}, we obtain that $\mathsf{Sig[\Pi]}$ is (s)UF-CMA.
\qed\end{proof}

Together with the fact that commit-and-open PCIPs can easily be made quantum extractable in the right sense by using standard hash-based commitments based on a collapsing hash function, we obtain the security of the MQDSS signature scheme. Recall that the standard hash-based commitment scheme works as follows. On input $s$, the commitment algorithm samples a random opening string $u$ and outputs it together with the commitment $c=H(s,u)$. Opening just works by recomputing the hash and comparing it with $c$ . Note that, while this commitment scheme is collapse-binding \cite{Unruh2016}, we need the stronger property of collapsingness of the function defined by the commitment algorithm that, on input a string and some randomness, outputs a commitment (collapse-binding only requires the collapsingness with respect to the committed string, not the opening information).
\begin{cor}[sUF-CMA of MQDSS]
	Let  $\mathsf\Pi_{\mathrm{SSH}}$ be the  5-round identification scheme from \cite{SSH} repeated in parallel a suitable number of times and instantiated with the standard hash-based commitment scheme using a collapsing hash function. Then the Fiat-Shamir signature scheme constructed from  $\mathsf\Pi_{\mathrm{SSH}}$ is sUF-CMA.
\end{cor}
\begin{proof}[sketch]
	In $\mathsf\Pi_{\mathrm{SSH}}$, the honest prover's first message consists of two commitments, and the second and final messages contain functions of  the strings committed to in the first message. This structure, together with the computational binding property (implied by the collapse binding property) of the commitments, immediately implies that $\mathsf\Pi_{\mathrm{SSH}}$ has computationally unique responses.  According to Corollary \ref{cor:SSH-PoK} in the appendix, $\mathsf\Pi_{\mathrm{SSH}}$ is a quantum proof of knowledge. It also has HVZK according to \cite{SSH}. Finally, the first message of $\mathsf\Pi_{\mathrm{SSH}}$ is clearly unpredictable. An application of Theorem \ref{thm:(s)UFCMA} finishes the proof.
\qed\end{proof}


\subsection{Sequential Or Proofs}

A second application of our multi-input version of the measure-and-reprogram result is to the OR-proof as introduced by Liu, Wei and Wong \cite{LWW04} and further analyzed by Fischlin, Harasser and Janson \cite{FHJ}. This is an alternative (non-interactive) proof for proving existence/knowledge of (at least) one of two witnesses without revealing which one, compared to the well known technique by Cramer, Damg\r{a}rd and Schoenmakers \cite{CDS94}. 

Formally, given two $\Sigma$-protocols $\mathsf{\Sigma}_0$, and $\mathsf{\Sigma}_1$, for languages ${\cal L}_0$, and ${\cal L}_1$, respectively, \cite{LWW04} proposes as a non-interactive proof for the OR-language ${\cal L}_{\vee} = \{ (x_0,x_1) \,:\, x_0 \!\in\! {\cal L}_0 \vee x_1 \!\in\! {\cal L}_1\}$ a quadruple $\pi_{\vee} = (a_0,a_1,z_0,z_1)$ such that
$$
V_{\vee}^H(x_0,x_1,\pi_{\vee})\! :=\! \bigl[V_0\bigl(x_0,a_0,H(1,x_0,x_1,a_1),z_0\bigr) \wedge V_1\bigl(x_1,a_1,H(0,x_0,x_1,a_0),z_1\bigr)\bigr]
$$
is satisfied. Fischlin et al. call this construction {\em sequential OR proof}. 
We emphasize that the two challenges $c_0$ and $c_1$ are computed ``over cross'', i.e., the challence $c_0$ for the execution of $\mathsf{\Sigma}_0$ is computed by hashing $a_1$, and vice versa. It is straightforward to verify that if $\mathsf{\Sigma}_0$ and $\mathsf{\Sigma}_1$ are special honest-verifier zero-knowledge, meaning that for any challenge $c$ and response $z$ one can efficiently compute a first message $a$ such that $(a,c,z)$ is accepted, then it is sufficient to be able to succeed in {\em one} of the two {\em interactive} protocols $\mathsf{\Sigma}_0$ and $\mathsf{\Sigma}_1$ in order to honestly produce such an OR-proof $\pi_{\vee}$. Thus, depending on the context, it is sufficient that one instance is in the corresponding language, or that the prover knows one of the two witnesses, to produce $\pi_{\vee}$. 
Indeed, if, say, $x_0 \in {\cal L}_0$ (and a witness $w_0$ is available), then $\pi_{\vee}$ can be produced as follows. Prepare $a_0$ according to $\mathsf{\Sigma}_0$, compute $c_1 := H(0,x_0,x_1,a_0)$ and simulate $z_1$ and $a_1$ using the special honest-verifier zero-knowledge property of $\mathsf{\Sigma}_1$ so that $V_1(x_1,a_1,c_1,z_1)$ is satisfied, and then compute the response $z_0$ for the challenge $c_0 := H(1,x_0,x_1,a_1)$ according to $\mathsf{\Sigma}_0$. 

On the other hand, intuitively one expects that one of the two instances must be true in order to be able to successfully produce a proof. Indeed, \cite{LWW04} shows security of the sequential OR in the (classical) ROM. \cite{FHJ} go a step further and show security in the (classical) {\em non-programmable} ROM. Here we show that our multi-input version of the measure-and-reprogram result (as a matter of fact the 2-input version) implies security in the QROM. 

\begin{thm}
There exists a black-box quantum polynomial-time interactive algorithm $\hat{\cal P}$, which first outputs a bit $b$ and two instances $x_0,x_1$, and in a second stage acts as an interactive prover that runs $\mathsf{\Sigma}_b$ on instance $x_b$, such that for any adversary $\cal A$ making $q$ queries to a uniformly random function $H$ and for any $x_0^\circ,x_1^\circ$: 
	\begin{align*}
	&\Pr\bigr[x_0 = x_0^\circ \,\wedge\, x_1 = x_1^\circ \,\wedge\, v_b = accept :(b,x_0,x_1,v_b) \leftarrow \langle\hat{\cal P}^{\cal A} , {\cal V}_b\rangle\bigl] 
	\\ &
	\geq \frac{1}{(2q+1)^4} \Pr_H\bigr[x_0 = x_0^\circ \,\wedge\, x_1 = x_1^\circ \,\wedge\, V^H_{\vee}(x_0,x_1,\pi_{\vee}) : (x_0,x_1,\pi_{\vee}) \leftarrow {\cal A}^H \bigl] \, .
	\end{align*} 
\end{thm}
As explained above, the execution $(b,x_0,x_1,v_b) \leftarrow \langle\hat{\cal P}^{\cal A} , {\cal V}_b\rangle$ should be understood in that $\hat{\cal P}^{\cal A}$ first outputs $x_0,x_1$ and $b$, and then it engages with ${\cal V}_b$ to execute $\mathsf{\Sigma}_b$ on instance $x_b$. Thus, the statement ensures that if ${\cal A}^H$ succeeds to produce a convincing proof $\pi_{\vee}$ then $\hat{\cal P}^{\cal A}$ succeeds to convincingly run $\mathsf{\Sigma}_0$ {\em or} $\mathsf{\Sigma}_1$ (with similar success probability), where it is up to $\hat{\cal P}^{\cal A}$ to choose which one it wants to do. 

Of course, the statement translates to the {\em static} setting where the two instances $x_0$ and $x_1$ are {\em fixed} and not produced by the dishonest prover.

\begin{proof}
The algorithm ${\cal A}$ fits well into the statement of Theorem \ref{thm:multiplemar} with the two extractable inputs $\tilde x_0 = (0,x_0,x_1,a_0)$ and $\tilde x_1 = (1,x_0,x_1,a_1)$. Thus, we can consider the 3-stage algorithm ${\cal S}$ ensured by Theorem \ref{thm:multiplemar}, which behaves as follows with at least the probability given by the right hand side of the claimed inequality. In the first stage, it outputs a permutation on the set $\{0,1\}$, which we represent by a bit $b \in \{0,1\}$ with $b=0$ corresponding to the identity permutation, as well as $\tilde x_b = (b,x_0,x_1,a_b)$. 
On input a random $\Theta_b = c_{1-b}$ (``locally'' chosen by $\hat{\cal P}$), ${\cal S}$ then outputs $\tilde x_{1-b} = (1-b,x_0,x_1,a_{1-b})$. Finally, on input a random $\Theta_{1-b} = c_b$ (provided by ${\cal V}_b$ as the challenge upon the first message $a_b$), ${\cal S}$ outputs $z_0,z_1$ so that $V_{\vee}$ is satisfied with the challenges $c_b$ and $c_{1-b}$, and thus in particular $V_b\bigl(x_b,a_b,c_b,z_b\bigr)$ is satisfied. 
This directly shows the existence of $\hat{\cal P}$ as claimed. 

\qed\end{proof}

	\section{Acknowledgement}

	We thank Dominque Unruh for hinting towards the possibility of the improved Theorem 2 (compared to [DFMS19]), see also Footnote 8, and 
	Andreas H\"ulsing
	for helpful discussions. 
%
%
CM was funded by a NWO VENI grant (Project No. VI.Veni.192.159). SF was partly supported by the EU Horizon 2020 Research and Innovation	Program Grant 780701 (PROMETHEUS). JD was funded by
ERC-ADG project 740972 (ALGSTRONGCRYPTO).

	\bibliographystyle{alpha}
	\bibliography{QROM}

	\begin{appendix}

		\section{Quantum extractability of q2 identification schemes}

		A class of identification schemes that is of particular interest are so-called q2-identification schemes. The NIST candidate signature scheme MQDSS, for example, is obtained from such an identification scheme via the multi-round Fiat-Shamir transformation from Definition \ref{def:mFS-sig} (with some additional strings included in the hash arguments). In this section, we will prove that a PCIP with a so-called ``q2 extractor'' \cite[Definition 4.6]{MQDSS} is a quantum proof of knowledge if it has an additional collapsingness property. This is necessary for its Fiat-Shamir transformation to  fulfill (s)UF-CMA in the QROM (for (s)UF-CMA in the ROM, the q2-extractor alone is sufficient \cite{MQDSS}).
		
		We begin by defining q2 identification schemes and their extractors. 
		\begin{defi}
			A 5-round identification scheme is a q2 identification scheme, if the second challenge is a single bit. A q2 identification scheme is called q2-extractable if there exists a polynomial-time algorithm that, on input four transcripts $t^{(i)}=(a^{(i)}_1,c^{(i)}_1,a^{(i)}_2, c^{(i)}_2,z^{(i)})$, $i=1,2,3,4$, such that
			\begin{align}
			\begin{split}
			c^{(1)}_1=c^{(2)}_1&\neq	c^{(3)}_1=c^{(4)}_1\ \mathrm{ and}\\
			c^{(1)}_2=c^{(3)}_3&\neq	c^{(2)}_2=c^{(4)}_2,
			\end{split}\label{eq:q2-ext}
			\end{align}
			outputs the secret key with non-negligible probability.
		\end{defi}
	For ease of exposition we have assumed that the different  challenges of a single PCIP come all from the same challenge space. A q2 identification scheme can be brought into this form by having the prover compute the second challenge by selecting the first bit of an augmented second challenge that is as large as the first one.
		For classical provers, four transcripts as required by the above definition can be obtained by straightforward rewinding. In the following, we show that, if the q2 identification scheme has an additional property similar to the quantum-computationally unique responses property introduced in \cite{DFMS19,LZ19}, then the existence of a q2 extractor implies that there exists a quantum extractor. This makes the scheme a quantum proof of knowledge. The argument follows the same lines as the one given in \cite{DFMS19} to prove that $t$-soundness and quantum-computationally unique responses imply the quantum proof-of-knowledge-property, which in turn is an extension of the result by Unruh for \sigps with perfect unique responses \cite{Unruh2012}.
		
		Recall the definition of a collapsing relation, \cite[Definition 23]{DFMS19}, a generalization of the notion of a collapsing hash function \cite{Unruh2016}. We define the notion of  collapsingness for interactive proof systems as follows:
		\begin{defi}
			A $(2n\!+\!1)$-round interactive proof system $\mathsf\Pi$ is called collapsing, if the relation $R_{\mathsf \Pi}:{\cal X}\times {\cal Y}\to \{0,1\}$ with ${\cal X}=\mathcal C^n\times {\cal A}_1$ and ${\cal Y}={\cal A}_2\times...\times {\cal A}_n\times{\cal Z}$ given by the verification predicate $V_{\mathsf \Pi}$ of  $\mathsf \Pi$ is collapsing from ${\cal X}$ to ${\cal Y}$.
		\end{defi}
		Note that for $n=1$, this notion of collapsingness coincides with the notion of quantum-computa-\switch{}{\linebreak}tionally unique responses from \cite{DFMS19}.
		
		Given a q2-identification scheme $\mathsf \Pi$, consider the following straightforward (first stage of a) quantum extractor $\cal E_{\mathsf \Pi}^\mathcal A$. The extractor runs the prover $\mathcal A$ using honestly sampled challenges  to obtain a first transcript $t^{(1)}$. Now it rewinds three times and reruns $\mathcal A$, each time with a fresh pair of challenges, chosen such as to obtain $t^{(i)}$, $i=2,3,4$ such that the four transcripts fulfill the conditions \eqref{eq:q2-ext}. For this extractor, we obtain the following
		
		\begin{thm}
			Let $\mathsf \Pi$ a q2-extractable q2-identification scheme that is also collapsing. Then the success probability of the extractor ${\cal E}_{\mathsf \Pi}^\mathcal A$ is lower-bounded in terms of the success probability of the prover $\cal A$ as
			\begin{equation}
			\Pr[{\cal E}_{\mathsf \Pi}^\mathcal A\ \mathrm{extracts}]\ge \left(\Pr\bigr[ v = accept :(x,v) \leftarrow \langle{\cal A} , {\cal V}_{\mathsf \Pi}\rangle\bigl] \right)^7
			\end{equation}
		\end{thm}
		
		The proof of this theorem is essentially the same as for Theorem 25 in \cite{DFMS19}, which is a slight modification of an argument from \cite{Unruh2012}.
		
		As a corollary, we obtain the fact that for q2 identification schemes, q2-extractability and collapsingness imply the quantum proof of knowledge property as defined in \cite{Unruh2012}.
		\begin{cor}
			Let $\mathsf \Pi$ a q2-extractable q2-identification scheme that is also collapsing. Then it is a quantum proof of knowledge.
		\end{cor}
		
		In particular, the 5-round identification scheme $\mathsf\Pi_{\mathrm{SSH}}$ from \cite{SSH} which is used to construct the post-quantum digital signature scheme MQDSS has these properties under plausible assumptions, namely that it is instantiated with the standard hash-based commitment scheme using a collapsing hash function \cite{Unruh2016} (see discussion towards the end of Section \ref{sec:sig}). For MQDSS, this is no additional assumption, as the Fiat-Shamir transformation uses the QROM anyway, and a quantum accessible random oracle is collapsing by~\cite{Unruh2016}.
		
		\begin{cor}\label{cor:SSH-PoK}
			If the 5-round identification scheme from \cite{SSH} is instantiated with the standard hash-based commitment scheme using a collapsing hash function, it is a quantum proof of knowledge.
		\end{cor}
		\begin{proof}[sketch]
			According to \cite{MQDSS}, $\mathsf\Pi_{\mathrm{SSH}}$ is a q2-extractable q2 identification scheme. In $\mathsf\Pi_{\mathrm{SSH}}$, the honest prover's first message consists of two commitments, and the second and final messages contain functions of  the strings commited to in the first message, and some opening information, respectively. Measuring a function of a register is equivalent to a partial computational basis measurement of that register. According to the the collapsing property of the  hash function, no efficient algorithm can distinguish whether the the committed string and the opening information are measured or not. This clearly implies the same indistinguishability for partial measurements of the string register, which implies that $\mathsf\Pi_{\mathrm{SSH}}$ is collapsing.
		\qed\end{proof}
		Note that the above proof works for any multi-round PCIP that has a similar commit-and-open structure.

	\end{appendix}
	
\end{document}